\documentclass[11pt,a4paper]{article}
\usepackage{amsmath,amsfonts,amssymb,amsthm,enumerate,graphicx,xcolor,url}
\usepackage{mathtools}
\usepackage{tikz,authblk}
\usepackage{bm}
\usepackage{multirow}
\usepackage{tabularx}
\usepackage{nameref}
\usepackage{graphicx} 
\usepackage[colorlinks=true, citecolor =magenta, linkcolor = red,pdfencoding=auto, psdextra]{hyperref}
\usepackage[left=2.00cm, right=2.00cm, top=2.00cm, bottom=2.00cm]{geometry}
\usepackage{orcidlink}
\usepackage{subfig}
\allowdisplaybreaks

\newtheorem{theorem} {{\textsf{Theorem}}}[section]
\newtheorem{proposition}[theorem]{{\textsf{Proposition}}}
\newtheorem{corollary}[theorem]{{\textsf{Corollary}}}

\newtheorem{remark}[theorem]{{\textsf{Remark}}}
\newtheorem{example}[theorem]{{\textsf{Example}}}
\newtheorem{lemma}[theorem]{{\textsf{Lemma}}}

\newcommand{\F}{\mathbb{F}}
	\date{}
\begin{document}
	\title{Optimal Pure Quantum $(r,\delta)$-LRCs from Euclidean and Hermitian Dual-Containing Cyclic Codes} 
\author{Astha Agrawal\orcidlink{0000-0002-4583-5613}, Shayan Srinivasa Garani\orcidlink{0000-0002-2459-1445}}
\markboth{}{qLRC} 
\maketitle
\vspace{-17mm}
	\begin{center}
	\noindent {\small Division of EECS, Indian Institute of Science, Bengaluru, 560012, India.}
	\end{center}
	\footnotetext[1]{{\em E-mail addresses:} \url{asthaagrawaliitd@gmail.com} (A. Agrawal)}
    \footnotetext[2]{{\em E-mail addresses:} \url{shayangs@iisc.ac.in} (S. S. Garani)}
\maketitle
\begin{abstract}
Locally recoverable codes (LRCs) combine global error correction with the ability to
repair a small number of erased coordinates by accessing only a limited number
of other coordinates. Motivated by their quantum counterparts, we construct
several families of optimal cyclic $(r,\delta)$-LRCs that
are Euclidean or Hermitian dual-containing.   In the Euclidean
case, we obtain four families of optimal dual-containing cyclic $(r,\delta)$-LRCs over $\mathbb{F}_q$ with a range of minimum distances extending beyond the local distance $\delta$. In the Hermitian case, we
derive analogous families over $\mathbb{F}_{q^2}$ that are Hermitian dual-containing, when
$(r+\delta-1)\mid(q^2-1)$.  In addition, we develop a distinct construction for the case $(r+\delta-1)\mid(q^2+1)$ using symmetric defining sets and odd $\delta$, which
yields optimal codes with minimum distances $\ell+2$, $\delta+2$,
$2\delta-2$, and $2\delta$. The Euclidean Calderbank-Shor-Steane (CSS) and Hermitian stabilizer
constructions then give corresponding quantum cyclic $(r,\delta)$-LRCs over~$\mathbb{F}_q$. Further, the resulting stabilizer codes are pure; they meet the relevant quantum Singleton-type bound, and are therefore
\textit{optimal}. We also provide a comparison with previous works, highlighting the parameter regimes and minimum distance ranges covered by our results that are not attained by existing constructions. Explicit examples illustrate the constructions and verify the
dual-containment and purity conditions.
\end{abstract}
\section{Introduction}\label{sec:introduction}
Locally recoverable codes (LRCs) were introduced to reduce the cost of repairing
failed storage nodes in distributed systems. In a classical code with locality
$r$, each erased coordinate can be reconstructed by accessing at most $r$ other
coordinates \cite{GopalanEtAl2012}. The more general notion of $(r,\delta)$-locality, introduced in~\cite{PrakashEtAl2012}, requires every coordinate to belong to a local code of length at most $r+\delta-1$ and minimum distance at least $\delta$. Consequently, any $\delta-1$ erasures within a local recovery group can be repaired without accessing the entire codeword. The central problem is to retain this local repair
capability without sacrificing global distance or rate. This trade-off is
quantified by the Singleton-type bound for $(r,\delta)$-LRCs \cite{PrakashEtAl2012}, and codes meeting
the bound with equality are called optimal. Numerous constructions of optimal LRCs have been developed using cyclic-code techniques \cite{ChenEtAl2018,FANG2020101650,TamoBarg2014,TamoEtAl2016}.

Cyclic codes are especially attractive in this setting. Their algebraic
structure permits compact implementation, and both locality and minimum
distance can be controlled through the placement of zeros of the generator
polynomial. The defining set approach to cyclic LRCs was developed in~\cite{TamoEtAl2016} and subsequently used in several constructions of optimal
cyclic LRCs
\cite{ChenEtAl2018,FANG2020101650}. 
In particular, Fang et al.~\cite{FANG2020101650} constructed optimal cyclic
$(r,\delta)$-LRCs of unbounded length with global minimum distance exceeding
the local distance $\delta$. However, their constructions are not designed
to satisfy dual-containment conditions, which are essential when 
classical codes are used to construct quantum CSS codes.

Quantum locally recoverable codes extend the notion of local erasure recovery
to quantum information. Their systematic study was initiated by Golowich and
Guruswami \cite{GolowichGuruswami2023}, who introduced quantum LRCs and
derived corresponding locality bounds and constructions. A quantum
$(r,\delta)$-LRC permits the recovery of up to $\delta-1$ erased qudits
within a recovery set containing at most $r+\delta-1$ qudits~\cite{GalindoEtAl2026}. More recently, several
constructions of quantum LRCs have been obtained from
dual-containing classical codes \cite{GalindoHernandoMatsumoto2026BCH,LiEtAl2025Hermitian,LuoEtAl2025,RajpurohitBhaintwal2026,SharmaRamkumarTamo2025,YangFuLu2026SmallFields}. In particular, Euclidean dual-containing
codes can be used through the CSS construction, whereas Hermitian
dual-containing codes over $\mathbb F_{q^2}$ give stabilizer codes over
$\mathbb F_q$. 

Recent works have considered cyclic constructions of quantum LRCs from several
different viewpoints. Luo et al.~\cite{LuoEtAl2025} obtained cyclic CSS
families for the case $\delta=2$, while Rajpurohit and Bhaintwal~\cite{RajpurohitBhaintwal2026} studied quantum $(r,\delta)$-LRCs arising from
Euclidean dual-containing cyclic codes. Galindo et al.~
\cite{GalindoHernandoMatsumoto2026BCH} employed BCH and homothetic-BCH codes
with Euclidean and Hermitian duality to construct  optimal pure quantum
$(r,\delta)$-LRCs. In these explicit BCH-based optimal families, the global
quantum minimum distance is essentially tied to the local distance
$\delta$. This motivates the search for cyclic quantum LRCs for which the
global minimum distance can be increased beyond $\delta$ while retaining
locality, dual containment, and Singleton-type optimality.

The simultaneous realization of these three properties is nontrivial.
A periodic defining set can guarantee $(r,\delta)$-locality, while additional
zeros can increase the global minimum distance. However, adding zeros may destroy dual containment. Our main approach is to
start with a periodic locality set and then carefully choose additional
cyclotomic cosets whose residues extend the required distance pattern while
preserving $Z\cap(-Z)=\varnothing$ in the Euclidean setting, and $Z\cap(-qZ)=\varnothing
$ in the Hermitian setting, where $Z$ is the defining set of a cyclic code. 
The Hermitian setting is particularly useful because it substantially enlarges
the available parameter range. When
$r+\delta-1\mid(q^2-1),
$
we obtain Hermitian analogues of the Euclidean constructions over
$\mathbb F_{q^2}$. Although the resulting classical  family parameters are analogous to Euclidean case, Hermitian dual containment
produces quantum stabilizer codes over $\mathbb F_q$ and allows parameter
regimes that are not accessible through the corresponding Euclidean
construction. Moreover, the case
$r+\delta-1\mid(q^2+1)$
exhibits a different defining set geometry.
This motivates a separate signed-residue construction and leads to further
families of Hermitian dual-containing cyclic LRCs. The principal contributions of this paper are summarized as follows.

\begin{itemize}
\item We construct four families of Euclidean dual-containing optimal cyclic
$(r,\delta)$-LRCs over $\mathbb{F}_q$. Starting from the periodic locality set,
we consider a block of consecutive zeros, a one-zero extension, a two-zero
extension, and a general extension within the prescribed distance
$\delta+2\leq d\leq 2\delta$. These constructions yield classical distances
$\ell+1$, $\delta+1$, $\delta+2$, and $d$, respectively.

\item For $r+\delta-1\mid(q^2-1)$, we develop Hermitian analogues of the
Euclidean families of  optimal cyclic $(r,\delta)$-LRCs over
$\mathbb{F}_{q^2}$. The condition for Hermitian dual containment is handled
directly. 

\item For $r+\delta-1\mid(q^2+1)$, we introduce a different signed-residue
construction for odd $n$ and odd~$\delta$. The identity
$q^2\equiv-1\pmod{r+\delta-1}$ naturally pairs opposite residues into
$q^2$-cyclotomic cosets. This leads to optimal Hermitian dual-containing codes
with minimum distances $\ell+2$, $\delta+2$, $2\delta-2$, and $2\delta$.

\item Applying the CSS and Hermitian stabilizer constructions to the above
dual-containing classical codes yield corresponding families of optimal pure quantum
cyclic $(r,\delta)$-LRCs over $\mathbb F_q$. 
\end{itemize}
The paper is organized as follows. In Section~\ref{sec:preliminaries}, we briefly review cyclic defining sets, Euclidean and Hermitian duality, classical and
quantum $(r,\delta)$-locality, and the relevant Singleton-type bound. In Section \ref{sec:euclidean-cyclic-lrc}, we propose four families of optimal pure quantum $(r,\delta)$-LRC over $\mathbb{F}_q$ based on CSS codes. Subsequently, we investigate a few families of optimal pure quantum cyclic $(r,\delta)$-LRCs based on Hermitian dual-containing constructions, first for
divisors of $q^2-1$ and then for divisors of $q^2+1$ in Section \ref{sec:hermitian-cyclic-lrc}. Each group of results is
accompanied by explicit examples and a summary table. In Section \ref{sec:comparison} we compare the proposed constructions with prior works.  Finally, Section~\ref{sec:conclusion} concludes the paper and
discusses several directions for future work.

\section{Preliminaries}\label{sec:preliminaries}

Throughout the paper, $q$ denotes a prime power and $\mathbb{F}_q$ is the
finite field with $q$ elements. For a positive integer $n$, write
$\mathbb{Z}_n=\mathbb{Z}/n\mathbb{Z}$. Unless stated otherwise, all sets, and congruences are considered modulo $n$, and we assume
$\gcd(n,q)=1$.

\subsection{Linear codes, duality, and cyclic defining sets}

An $[n,k,d]_q$ linear code $C$ is a $k$-dimensional subspace of
$\mathbb{F}_q^n$ with minimum Hamming distance $d$, where $d=\min\{\operatorname{wt}(\bm x): \bm{0}\neq \bm x\in C\}$. Its Euclidean dual is defined as
$$
C^{\perp}=\{(x_0,x_1,\dots,x_{n-1})\in\mathbb{F}_q^n:
\sum_{i=0}^{n-1}x_i c_i=0\text{ for every }(c_0,c_1,\dots,c_{n-1})\in C\}.
$$
The code $C$ is called Euclidean dual-containing if $C^{\perp}\subseteq C$. The Hermitian dual of a code $C$ over $\mathbb{F}_{q^2}$ is defined as 
$$
C^{\perp_H}=\{(x_0,x_1,\dots,x_{n-1})\in\mathbb{F}_{q^2}^n:
\sum_{i=0}^{n-1}x_i c_i^q=0\text{ for every }(c_0,c_1,\dots,c_{n-1})\in C\}.
$$
The code $C$ is Hermitian
dual-containing if $C^{\perp_H}\subseteq C$.

A linear code $C$ of length $n$ is cyclic  if, for every codeword
$(c_0, c_1,\dots,c_{n-1}) \in C$, its cyclic shift $(c_{n-1},c_0,\dots,c_{n-2})$ also belongs to $C$. Equivalently, after identifying a vector $(c_0, c_1,\dots,c_{n-1})$ with its associated polynomial $c_0+c_1x+\dots+c_{n-1}x^{n-1}$,
$C$ is an ideal of $\mathbb{F}_q[x]/\langle x^n-1\rangle$. Let $\alpha$ be a
primitive $n$th root of unity in a suitable extension field. The complete
defining set of a cyclic code $C$ is defined as
$$
Z(C)=\left\{i\in\mathbb{Z}_n:g(\alpha^i)=0\right\},
$$
where $g(x)$ is the generator polynomial of $C$.
The defining set $Z(C)$ is a union of $q$-cyclotomic cosets modulo $n$. More precisely, for $i\in\mathbb{Z}_n$, the $q$-cyclotomic coset containing $i$ is \[ C_i=\{iq^j:j\geq0\}\pmod{n}. \]  Thus, $Z(C)=\bigcup_{\ell=1}^{s} C_{i_\ell}$ for some representatives $i_1,\ldots,i_s$. For a cyclic code over $\mathbb{F}_{q^2}$, the defining set is instead a union of $q^2$-cyclotomic cosets modulo $n$. Moreover, $ \dim(C)=n-|Z(C)|.$ These standard properties of cyclic codes and their defining sets can be found in \cite{HuffmanPless2003}.

For $A\subseteq\mathbb{Z}_n$, define $
-A=\{-a\pmod n:a\in A\},\ -qA=\{-qa\pmod n:a\in A\}.$
The following defining set characterizations \cite{HuffmanPless2003,AlyKlappeneckerSarvepalli2007} will be used repeatedly  :
$$
C^{\perp}\subseteq C
\quad\Longleftrightarrow\quad
Z(C)\cap(-Z(C))=\varnothing,
$$
and, for a cyclic code over $\mathbb{F}_{q^2}$,
$$
C^{\perp_H}\subseteq C
\quad\Longleftrightarrow\quad
Z(C)\cap(-qZ(C))=\varnothing.
$$
We will also use the following BCH bound and its generalized form \cite{BoseRayChaudhuri1960,Hocquenghem1959}.
\begin{lemma}[BCH bound]
     Let $C$ be an
$[n,k,d]_q$ cyclic code with defining set $Z$. Let m be a positive
integer such that $\gcd(m,n)=1$ and let $u,b\geq0$ be integers.
If the set $\{u+m(b+i) : i\in[0,\lambda-2]\}$ is contained in $Z$, then $d\geq \lambda$.
\end{lemma}
This observation, together with the defining sets of the Euclidean and Hermitian duals, yields the following useful lemma.
\begin{lemma}\label{dual distance-complement}
Let $C$ be a cyclic code of length $n$ with complete defining set
$Z\subseteq \mathbb Z_n$, and write $Z^c=\mathbb Z_n\setminus Z$.
If $C$ is a cyclic code over $\mathbb F_q$ (respectively, over $\mathbb{F}_{q^2}$) and $Z^c$ contains $t$ consecutive integers then $d(C^\perp)\ge t+1$ (respectively, $d(C^{\perp_H})\ge t+1$).
\end{lemma}
\begin{proof}
For the Euclidean dual, the defining set is $Z(C^\perp)=-Z^c.$ Hence, if $    a,a+1,\ldots,a+t-1\in Z^c,$ then $Z(C^\perp)$ contains $-a,-a-1,\ldots,-a-(t-1),$ which are $t$ consecutive defining zeros. The BCH bound gives
$d(C^\perp)\ge t+1$. For the Hermitian dual, we have $Z(C^{\perp_H})=-qZ^c.$ Thus $Z(C^{\perp_H})$ contains $-qa,-q(a+1),\ldots,-q(a+t-1),$ which form an arithmetic progression of length $t$ and common difference $-q$. Since $\gcd(q,n)=1$, the  BCH bound yields $d(C^{\perp_H})\ge t+1.$
\end{proof}
More generally, the same conclusions hold if the corresponding elements
of $Z^c$ form an arithmetic progression with common difference $b$
satisfying $\gcd(b,n)=1$.

\subsection{Classical $(r,\delta)$-locally recoverable codes}

Let $C$ be an $[n,k,d]_q$ linear code. For $J\subseteq\{1,\ldots,n\}$, let
$C|_J$ denote the puncturing of $C$ to the coordinates in $J$, i.e., $C|_J=\{ (c_j)_{j\in J} : (c_0,c_1,\ldots,c_{n-1})\in C \}$. The code $C$
has  $(r,\delta)$-locality if, for every coordinate $i$, there exists
a set $J_i$ containing $i$ such that
$$
|J_i|\leq r+\delta-1
\qquad\text{and}\qquad
d(C|_{J_i})\geq\delta.
$$
Thus, any set of at most $\delta-1$ erasures inside $J_i$ can be recovered
using only the remaining coordinates in that local set. The case $\delta=2$ reduces to the usual notion of locality $r$. The parameters of an $[n,k,d]$  linear code with $(r,\delta)$-locality satisfy the
Singleton-type bound \cite{PrakashEtAl2012}
\begin{equation}
d\leq n-k+1-
\left(\left\lceil\frac{k}{r}\right\rceil-1\right)(\delta-1).
\label{singleton_lrc}
\end{equation}
An $(r,\delta)$-LRC attaining equality in
\eqref{singleton_lrc} is called \textit{optimal}.
Define the locality set
\[
L=\bigcup_{i=1}^{\delta-1}L_{\ell_i}=\bigcup_{i=1}^{\delta-1}
\left\{m(r+\delta-1)+\ell_i:0\leq m\leq u-1\right\}.
\]
We recall the following two defining set criteria that will be used throughout the paper.

\begin{proposition}\cite[Proposition 2]{TianZhu2022}\label{(r,delta)locality}
Suppose that $\gcd(n,q)=1$ and $(r+\delta-1)\mid n$. Let $C$ be a cyclic code of length $n$ with complete defining set $Z$. Let $\ell_1<\ell_2<\cdots<\ell_{\delta-1}$ be an arithmetic progression with common difference $b$ satisfying $\gcd(b,n)=1.$  If $L\subseteq Z$
then $C$ has $(r,\delta)$-locality.
\end{proposition}
\begin{lemma}\cite[Lemma 2]{TianZhu2022}\label{distance}
Let $C$ be a cyclic code as defined in Proposition~\ref{(r,delta)locality}, and let $d'\leq 2\delta$. Suppose that there exists a set $D=\{t_1,t_2,\ldots,t_{d'-1}\}\subseteq Z$ such that their residues modulo $r+\delta-1$ form an
arithmetic progression with common difference $a$,  satisfying $\gcd(a,r+\delta-1)=1$. Then the minimum distance of $C$ satisfies $d(C)\geq d'$.
\end{lemma}



\subsection{Quantum $(r,\delta)$-locally recoverable codes}

A $q$-ary quantum code with parameters
$\bigl[\!\bigl[n,\kappa,d_\mathcal{Q}\bigr]\!\bigr]_q$
is a $q^{\kappa}$-dimensional subspace of
$(\mathbb{C}^{q})^{\otimes n}$ that detects every error of weight less than
$d_\mathcal{Q}$. It is a quantum $(r,\delta)$-LRC if, for every
coordinate $1\leq i\leq n$, there exists a set $J_i$ containing $i$, with
$|J_i|\leq r+\delta-1$, such that erasures on any $\delta-1$ coordinates of $J_i$ can be reversed by a recovery operation supported on $J_i$
\cite{GalindoEtAl2026}. We use the following two standard stabilizer constructions
\cite{CalderbankShor1996,Steane1996,AshikhminKnill2001} throughout the paper.

\begin{itemize}
\item If $C$ is an $[n,k,d]_q$ Euclidean dual-containing code, then the CSS construction gives a stabilizer code $\mathcal{Q}(C)$ with parameters
$
\bigl[\!\bigl[n,2k-n,d_\mathcal{Q}\bigr]\!\bigr]_q,
\text{ where }
d_\mathcal{Q}=\min\{\operatorname{wt}(\bm{c}):\bm{c}\in C\setminus C^{\perp}\}.
$

\item If $C$ is an $[n,k,d]_{q^2}$ Hermitian dual-containing code, then the Hermitian construction gives a stabilizer code $\mathcal{Q}(C)$ with parameters
$
\bigl[\!\bigl[n,2k-n,d_\mathcal{Q}\bigr]\!\bigr]_q,
\text{ where }
d_\mathcal{Q}=\min\{\operatorname{wt}(\bm{c}):\bm{c}\in C\setminus C^{\perp_H}\}.
$
\end{itemize}
A quantum code $\mathcal{Q}(C)$ is said to be \textit{pure} if its minimum distance $d_{\mathcal{Q}}$ is equal to
the minimum distance of the code $C$.

A Euclidean dual-containing code over $\mathbb{F}_q$, or a Hermitian dual-containing code over $\mathbb{F}_{q^2}$, gives rise to a quantum stabilizer code over $\mathbb{F}_q$. Under a mild condition on the minimum distance of the corresponding dual code, the classical and quantum notions of $(r,\delta)$-local recoverability are equivalent. We recall the following results from \cite{GalindoEtAl2026}, which forms the principal bridge between the classical constructions developed in this paper and the resulting quantum locally recoverable codes. 
\begin{theorem}\cite[Theorem 29]{GalindoEtAl2026}\label{classical-quantum-locality}
Let $C$ be  a linear code of length $n$ over $\mathbb{F}_{q^2}$ (respectively, over $\mathbb{F}_q$) satisfying $ C^{\perp_H}\subseteq C$ ( respectively, $C^{\perp}\subseteq C)$. Let $\mathcal{Q}(C)$ denote the associated quantum stabilizer code over $\mathbb{F}_q$. Assume that $ \delta\leq d\bigl(C^{\perp_H}\bigr) $ ( respectively, $\delta\leq d\bigl(C^{\perp}\bigr). $ Then, $\mathcal{Q}(C)$ is a quantum $(r,\delta)$-LRC if and only if $C$ is a classical $(r,\delta)$-LRC. \end{theorem}

\begin{remark}
For all dual-containing constructions considered in this paper, the dual distance requirement $( \delta\leq d(C^{\perp_H}) \ \text{or}\ \delta\leq d(C^{\perp}) )$ follows directly. Since 
$
C^{\perp}\subseteq C,\  d(C^{\perp})\geq d(C)
\ \text{or }\
C^{\perp_H}\subseteq C,\  d(C^{\perp_H})\geq d(C)
$
together with $d(C)\geq\delta$. As all of our constructions satisfy \(d(C)\geq\delta\), therefore, we do not verify this condition separately in each construction.
\end{remark}
We next recall the Singleton-type bound for quantum $(r,\delta)$-LRCs
obtained from dual-containing classical codes \cite{GalindoEtAl2026}.

\begin{theorem}\cite[Theorem 31]{GalindoEtAl2026}
\label{thm:singleton_qlrc}
Let $C$ satisfy the hypotheses of
Theorem~\ref{classical-quantum-locality} and suppose that
$
\dim(C)=\frac{n+\kappa}{2}.
$ If $C$ is a classical $(r,\delta)$-LRC, then the associated quantum $(r,\delta)$-LRC  has parameters
$
\bigl[\!\bigl[n,\kappa,d_{\mathcal Q}\geq d(C)\bigr]\!\bigr]_q
$
and satisfies
$$
\kappa+2d(C)
+2\left(
\left\lceil\frac{n+\kappa}{2r}\right\rceil-1
\right)(\delta-1)
\leq n+2.
$$
\end{theorem}

This  bound is expressed in terms of the minimum distance of the
classical code $C$. If $\mathcal{Q}(C)$ is pure, then
$d_{\mathcal Q}=d(C)$, and the bound becomes
\begin{equation}
 \kappa+2d_{\mathcal Q}
+2\left(
\left\lceil\frac{n+\kappa}{2r}\right\rceil-1
\right)(\delta-1)
\leq n+2.  
\label{singleton_qlrc}
\end{equation}

A pure quantum $(r,\delta)$-LRC is called \emph{optimal} if it attains the
above Singleton-type bound \eqref{singleton_qlrc} with equality.
In particular, if the underlying classical $(r,\delta)$-LRC is optimal and
$d(C^{\perp})>d(C)$, or respectively
$d(C^{\perp_H})>d(C)$, then the associated quantum code is pure, has
$d_Q=d(C)$, and attains~\eqref{singleton_qlrc} with equality.

\section{Euclidean dual-containing cyclic $(r,\delta)$-LRCs}
\label{sec:euclidean-cyclic-lrc}
In this section, we construct four main families of Euclidean dual-containing optimal cyclic $(r,\delta)$-LRCs. As a consequence, these constructions yield four corresponding families of pure quantum $(r,\delta)$-LRCs that are optimal with respect to the Singleton-type bound \eqref{singleton_qlrc}. Throughout this section, let  $n=u(r+\delta-1)$ be a positive integer  where $u\geq 2$, $r\geq 1$, $\delta\geq 2$ and $\gcd(n,q)=1$. In this section, we consider the locality set  $$L=\bigcup_{i=1}^{\delta-1}L_i=\bigcup_{i=1}^{\delta-1}\{m(r+\delta-1)+i:0\leq m\leq u-1\}.$$ The residues modulo $(r+\delta-1)$ gives $1,2,\dots,\delta-1$, which is in arithmetic progression of length $\delta-1$ and common difference $1$. Hence, by Proposition \ref{(r,delta)locality}, $L$ defines a valid locality set with locality $(r,\delta)$.

The next elementary lemma collects the two facts about $L$ that are used in every construction of this section.
\begin{lemma}\label{Lqcyclotomic}
Suppose that $(r+\delta-1)\mid(q-1)$. Then, $L$ is a union of $q$-cyclotomic cosets modulo $n$. Moreover, if $r>\delta-1$, then  $L\cap(-L)=\varnothing.$
\end{lemma}

\begin{proof}
Since $(r+\delta-1)\mid(q-1)$,  there exists an integer $a$ such that $q=1+a(r+\delta-1)$ then for every $0\leq m\leq u-1$ and
$1\leq i\leq\delta-1$, we have
\[
q(m(r+\delta-1)+i)\equiv m'(r+\delta-1)+i\pmod{n}.
\]
where $0\leq m'\leq u-1$. Thus, each $L_i$ is a union of $q$-cyclotomic cosets modulo $n$.
Consequently, $L$ is also a union of $q$-cyclotomic cosets modulo $n.$
Suppose that $L\cap(-L)\neq\varnothing$. Then there exist
$0\leq m_1,m_2\leq u-1$ and $1\leq i_1,i_2\leq\delta-1$ such that
\[
m_1(r+\delta-1)+i_1\equiv -m_2(r+\delta-1)-i_2\pmod{n}.
\]
Reducing the above congruence modulo $(r+\delta-1)$ gives
\[
i_1+i_2\equiv0\pmod{(r+\delta-1)}.
\]
However, $2\leq i_1+i_2\leq 2(\delta-1)<r+\delta-1,
$
where the last inequality follows from $r>\delta-1$. This is a
contradiction. Therefore, $L\cap(-L)=\varnothing.$
\end{proof}
\begin{remark}
\label{rem:symmetry}
For arbitrary subsets $A,B\subseteq\mathbb{Z}_n$,
$$
A\cap(-B)=\varnothing
\quad\Longleftrightarrow\quad
B\cap(-A)=\varnothing.
$$
Thus, in every dual-containment argument below, only one of the two symmetric
cross-intersections needs to be checked.
\end{remark}
\begin{theorem}\label{thm:basic construction}
Let $n=u(r+\delta-1)$ be a positive integer such that 
$n\mid(q-1)$ and $\delta\leq \ell<r.$ Then, there exists a dual-containing optimal cyclic $(r,\delta)$-LRC with parameters $\bigl[u(r+\delta-1),\, ur-\ell+(\delta-1),\, \ell+1\bigr]_q.$ Moreover, there exists an optimal pure  quantum cyclic $(r,\delta)$-LRC with parameters $\bigl[\!\bigl[
    u(r+\delta-1),\,
    u(r-\delta+1)-2\ell+2(\delta-1),\,
     \ell+1
    \bigr]\!\bigr]_q.$
\end{theorem}
\begin{proof}
Let $ D=\{1,2,\ldots,\ell\},$ where $\delta\leq \ell<r$  and set $Z=L\cup D$. Since
$n\mid(q-1)$, every $q$-cyclotomic coset modulo $n$ is a singleton. Hence, $Z$ is the complete defining set of a cyclic code
$C$ over $\mathbb{F}_q$ and  by Proposition \ref{(r,delta)locality}, $C$ has $(r,\delta)$ locality. We first show that $C$ is dual-containing. It suffices to prove that $Z\cap(-Z)=\varnothing$. Note that $$-L=\bigcup_{i=1}^{\delta-1}\{m(r+\delta-1)-i:0\leq m\leq u-1\}\ \text{ and }\ -D=\{n-j:1\leq j\leq\ell\},$$ where all elements are considered
modulo $n$. By Lemma \ref{Lqcyclotomic}, we have $L\cap-L=\varnothing$, as $\delta\leq \ell<r$.  Next, suppose that $x\in L\cap(-D)$. Then, for some
$0\leq m_1\leq u-1$, $1\leq i_1\leq\delta-1$, and
$1\leq j_1\leq\ell$,
\[
m_1(r+\delta-1)+i_1\equiv-j_1\pmod{u(r+\delta-1)}.
\]
Consequently, $i_1+j_1\equiv0\pmod{r+\delta-1}$. On the other hand,
$2\leq i_1+j_1\leq\ell+\delta-1<r+\delta-1$, since $\ell<r$.
Thus, $L\cap(-D)=\varnothing$. Finally, suppose that $x\in D\cap(-D)$. Then, for some
$1\leq \ell_1,\ell_2\leq\ell$,
$$\ell_1\equiv-\ell_2\pmod{u(r+\delta-1)},$$ and hence
$\ell_1+\ell_2\equiv0\pmod{u(r+\delta-1)}$. Since
$2\leq\ell_1+\ell_2\leq2\ell<2r<n$, this is impossible. Therefore,
$D\cap(-D)=\varnothing$. Moreover,
$D\cap(-L)=\varnothing$ follows equivalently from
$L\cap(-D)=\varnothing$. Hence,
$Z\cap(-Z)=\varnothing$, and therefore
$C^{\perp}\subseteq C$. We now determine the parameters of $C$. Since
$D$ already contains the elements $1,\ldots,\delta-1$ corresponding to
the $m=0$ part of $L$, i.e., $|L\cap D|=\delta-1$, we have
$|Z|=|L|+|D|-|L\cap D|=(u-1)(\delta-1)+\ell$. Therefore,
\[
\dim(C)
=n-|Z|
=u(r+\delta-1)-(u-1)(\delta-1)-\ell
=ur-\ell+(\delta-1).
\]
Furthermore, $Z$ contains the $\ell$ consecutive integers
$1,2,\ldots,\ell$. Hence, the BCH bound gives
$d(C)\geq\ell+1$. On the other hand, applying
bound \eqref{singleton_lrc} yields $d(C)\leq\ell+1$. Thus,
$d(C)=\ell+1$. Consequently, $C$ is a dual-containing optimal cyclic
$(r,\delta)$-LRC with parameters
\[
\bigl[u(r+\delta-1),\,ur-\ell+(\delta-1),\,\ell+1\bigr]_q.
\]
Hence, Theorem \ref{classical-quantum-locality} ensures there exists a quantum $(r,\delta)$-LRC. We finally show that the resulting quantum code is pure. Since
$D=\{1,2,\ldots,\ell\}$ affects only the first block and $u\geq2$,
for any $1\leq m\leq u-1$, the $r$ consecutive integers
\[
    m(r+\delta-1)+\delta,\,
    m(r+\delta-1)+\delta+1,\ldots,
    (m+1)(r+\delta-1)
\]
belong to $Z^c$. Hence, by Lemma~\ref{dual distance-complement}, $d(C^\perp)\ge r+1.$
Since $\ell<r$ and $d(C)=\ell+1$, we obtain $d(C^\perp)\ge r+1>\ell+1=d(C).$
Therefore, the quantum code is pure, and its minimum distance is $d_\mathcal{Q}=d(C)=\ell+1$. Consequently, by CSS construction, there exists a pure quantum cyclic $(r,\delta)$-LRC with
parameters
\[
\bigl[\!\bigl[
u(r+\delta-1),\,
u(r-\delta+1)-2\ell+2(\delta-1),\,
\ell+1
\bigr]\!\bigr]_q
\]
and is optimal w.r.t the quantum Singleton-type bound
\eqref{singleton_qlrc}.
 \end{proof}
\begin{example}
Let $q=29$, $r=5$, $\delta=3$, $\ell=4$, and $u=2$. Then $n=u(r+\delta-1)=2(5+3-1)=14,$ and $14\mid 28=q-1,\ \delta\leq \ell<r$ Thus, all the hypotheses of Theorem~\ref{thm:basic construction} are satisfied. More explicitly, the defining set of the corresponding cyclic code is $Z=L\cup D,$ where
$$
L=\bigcup_{i=1}^{2}
\{7m+i:0\leq m\leq 1\}
=\{1,2,8,9\} \text{ and } 
D=\{1,2,3,4\}.
$$
Consequently, $Z=\{1,2,3,4,8,9\}.$ Moreover, $
-Z=\{5,6,10,11,12,13\}\pmod {14},$ and hence $Z\cap(-Z)=\varnothing.$ Therefore, the resulting cyclic code $C$ is Euclidean dual-containing. Using Magma \cite{magma}, we verified $d(C)=5$ and $d(C^{\perp})=8\neq d(C)$. Hence, there exists a dual-containing optimal
cyclic $(5,3)$-LRC over $\mathbb{F}_{29}$ with parameters
$[14,8,5]_{29}.$
and the resulting quantum code is pure. Moreover, there exists an optimal quantum cyclic
$(5,3)$-LRC over $\mathbb{F}_{29}$ with parameters
$
\left[\!\left[14,2,5\right]\!\right]_{29}.
$
\end{example}
The next result no longer requires $n|(q-1)$. Instead, we  construct a dual-containing cyclic $(r,\delta)$-LRC by adding one suitably chosen defining zero to the locality defining set.
\begin{theorem}\label{thm:one zero extension}
Let $n=u(r+\delta-1)$ be a positive integer. Assume that $(r+\delta-1)\mid(q-1)$ and  $2\leq\delta< r$. Let $\eta=\gcd(u,q-1)$. Suppose that at least one of the following conditions holds: 
\begin{enumerate}
    \item $\eta>2$;
    \item $\eta\leq 2$, $r>\delta+1$,and $\eta\mid\dfrac{\delta(q-1)}{r+\delta-1}$.
\end{enumerate}
Then, there exists a dual-containing
optimal cyclic $(r,\delta)$-LRC with parameters
$\bigl[u(r+\delta-1),\,ur-1,\,\delta+1\bigr]_q.$
Moreover, there exists an optimal pure quantum cyclic
$(r,\delta)$-LRC with parameters
$\bigl[\!\bigl[u(r+\delta-1),\,u(r-\delta+1)-2,\,
\delta+1\bigr]\!\bigr]_q.$
\end{theorem}
\begin{proof}
We consider the two cases separately.

\medskip
\noindent\textbf{Case 1.} Suppose that $\eta=\gcd(u,q-1)>2$, and set $x=u/\eta$. Clearly,
$1\leq x\leq u-1$. Since $\eta\mid(q-1)$, we have
$$(q-1)x(r+\delta-1)\equiv0\pmod{u(r+\delta-1)}.$$ Thus,
$q x(r+\delta-1)\equiv x(r+\delta-1)\pmod{n}$; hence, the
$q$-cyclotomic coset containing $x(r+\delta-1)$ is a singleton. Define
$$D=\{x(r+\delta-1)\}.$$ Hence, $Z=L\cup D$ is the complete defining set of a cyclic code $C$ over $\mathbb{F}_q$ and by Proposition \ref{(r,delta)locality}, $C$ has $(r,\delta)$ locality. Now, we will show that $C$ is dual-containing. Since \(\eta>2\), we have \(\eta\nmid 2\); hence, 
\[
2x(r+\delta-1)
=2\frac{u}{\eta}(r+\delta-1)
\not\equiv 0
\pmod{u(r+\delta-1)}.
\]
Thus, $x(r+\delta-1)
\not\equiv -x(r+\delta-1)
\pmod{u(r+\delta-1)},$ which implies $D\cap(-D)=\varnothing.$ Next, we show that $L\cap(-D)=\varnothing$. Suppose, to the contrary,
that some element belongs to $L\cap(-D)$. Then, for some
$0\leq m\leq u-1$ and $1\leq i\leq\delta-1$, we have 
$$m(r+\delta-1)+i\equiv-x(r+\delta-1)\pmod{u(r+\delta-1)}.$$
Reducing this congruence modulo $(r+\delta-1)$ gives
$$i\equiv0\pmod{r+\delta-1},$$ which is impossible since
$1\leq i\leq\delta-1<r+\delta-1$. Hence,
$L\cap(-D)=\varnothing$. Equivalently, $(-L)\cap D=\varnothing$. Together with $L\cap(-L)=\varnothing$, the preceding observations yield
$Z\cap(-Z)=\varnothing$. Hence,
$C^{\perp}\subseteq C$. Also,  the complete defining set $Z$
contains the $\delta$ consecutive integers
$x(r+\delta-1),x(r+\delta-1)+1,\ldots,
x(r+\delta-1)+\delta-1$.
Therefore, the BCH bound gives $d(C)\geq\delta+1$. It remains to determine the dual distance. Since $\eta>2$, we have
$u\geq3$.  Therefore, there exists an integer $0\leq m\leq u-1$ such that $m+1\neq x$ for which
\[
    m(r+\delta-1)+\delta,\ \,
    m(r+\delta-1)+(\delta+1),\ldots,
    (m+1)(r+\delta-1)
\]
are all contained in $Z^c$. This is a set of $r$ consecutive integers.
Hence, by Lemma~\ref{dual distance-complement}, $d(C^\perp)\ge r+1.$ Since $r>\delta$,
$d(C^\perp)\ge r+1>\delta+1.$

\medskip
\noindent\textbf{Case 2.} Suppose that $\eta\leq 2$, $r>\delta+1$, and
$\eta\mid\dfrac{\delta(q-1)}{r+\delta-1}$. We first show that there exists an integer $x$ such that
$D=\{x(r+\delta-1)+\delta\}$ is a singleton $q$-cyclotomic coset
modulo $n=u(r+\delta-1)$. For this, it is enough to find $x$ satisfying
$$(q-1)(x(r+\delta-1)+\delta)\equiv0\pmod{u(r+\delta-1)}.$$
Since $(r+\delta-1)\mid(q-1)$, this congruence is equivalent to
$$(q-1)x\equiv-\dfrac{\delta(q-1)}{r+\delta-1}\pmod{u}.$$
By the solvability criterion for linear congruences, such an $x$ exists
if and only if
$\eta\mid\dfrac{\delta(q-1)}{r+\delta-1}$. It follows directly from the assumption. Hence, there
exists an integer $x$ such that $D$ is a singleton $q$-cyclotomic coset
modulo $n$. We claim that $D\cap(-D)=\varnothing$. Otherwise,
$$x(r+\delta-1)+\delta\equiv
-\bigl(x(r+\delta-1)+\delta\bigr)\pmod{n},$$ which, upon reduction
modulo $(r+\delta-1)$, yields
$2\delta\equiv0\pmod{r+\delta-1}$. However, $r>\delta+1$ implies
$0<2\delta<r+\delta-1$, a contradiction. Hence,
$D\cap(-D)=\varnothing$. We now show that $L\cap(-D)=\varnothing$. Suppose, for some $0\leq m\leq u-1$ and
$1\leq i\leq\delta-1$,
$$m(r+\delta-1)+i
\equiv-\bigl(x(r+\delta-1)+\delta\bigr)\pmod{n}.$$
Reducing modulo $(r+\delta-1)$ yields
$i+\delta\equiv0\pmod{r+\delta-1}$.
On the other hand,
$0<i+\delta\leq2\delta-1<r+\delta-1$, where the last inequality follows
from $r>\delta+1$. This is a contradiction. Hence,
$L\cap(-D)=\varnothing$. Equivalently,
$(-L)\cap D=\varnothing$. Let $C$ be the cyclic code with complete defining set $Z=L\cup D$.
Together with $L\cap(-L)=\varnothing$, the above observations imply
$Z\cap(-Z)=\varnothing$. Therefore,
$C^{\perp}\subseteq C$. Moreover, by the construction of $L$, the code
$C$ has $(r,\delta)$-locality. Since $x(r+\delta-1)+\delta\in D$ and
$x(r+\delta-1)+i\in L$ for $1\leq i\leq\delta-1$, the defining
set $Z$ contains the $\delta$ consecutive integers.
Hence, the BCH bound gives $d(C)\geq\delta+1$. The additional defining zero set $D$ has residue $\delta$ modulo $r+\delta-1$. Hence, for every unaffected block, the integers
\[
    m(r+\delta-1)+(\delta+1),m(r+\delta-1)+(\delta+2),\ldots,
    (m+1)(r+\delta-1)
\]
belong to $Z^c$. There are $r-1$ consecutive integers in this set, and
therefore $d(C^\perp)\ge r.$ Since $r>\delta+1$, $d(C^\perp)\ge r>\delta+1.$

\medskip
 In both cases, $D\cap L=\varnothing$, and hence
$|Z|=u(\delta-1)+1$. Thus,
$$\dim(C)=n-|Z|=u(r+\delta-1)-u(\delta-1)-1=ur-1.$$
By bound~\eqref{singleton_lrc}, $d(C)\leq\delta+1$.
Consequently, $d(C)=\delta+1$, and $C$ is an
optimal dual-containing cyclic $(r,\delta)$-LRC with parameters
$\bigl[u(r+\delta-1),\,ur-1,\,\delta+1\bigr]_q.$
Therefore, by Theorem~\ref{classical-quantum-locality}, $\mathcal{Q}(C)$ is a quantum $(r,\delta)$-LRC. Since, in both cases, $d(C^\perp)>d(C)=\delta+1$, the code $\mathcal{Q}(C)$ is pure with $d_{\mathcal Q}=\delta+1$. Hence, it has parameters
$
\bigl[\!\bigl[u(r+\delta-1),u(r-\delta+1)-2,\delta+1\bigr]\!\bigr]_q.
$
Moreover, it attains the Singleton-type bound~\eqref{singleton_qlrc}, and is therefore optimal.
\end{proof}
\begin{example} Let $q=13$, $u=4$, $r=8$, and $\delta=5$. Then $r+\delta-1=12\mid12=q-1$ and $\delta\leq r$. Moreover,
$\eta=\gcd(u,q-1)=\gcd(4,12)=4>2.$ Thus, all the hypotheses of Case~(1) of
Theorem~\ref{thm:one zero extension} are satisfied. Notice that
$n=u(r+\delta-1)=4\cdot12=48.$ The locality set is
\begin{align*}
L
&=\bigcup_{i=1}^{4}\{12m+i:0\leq m\leq3\}=\{1,2,3,4,13,14,15,16,25,26,27,28,37,38,39,40\}.
\end{align*}
Since $\eta=4$, we take $x=u/\eta=1$. Thus, $D=\{x(r+\delta-1)\}=\{12\}.$ Therefore, the complete defining set of a cyclic code $C$ over $\mathbb{F}_{13}$ is $$Z=L\cup D=\{1,2,3,4,12,13,14,15,16,25,26,27,28,37,38,39,40\}.$$
The negative of $Z$ modulo $48$ is $$-Z=\{8,9,10,11,20,21,22,23,32,33,34,35,36,44,45,46,47\}.$$ Hence, $Z\cap-Z=\varnothing$. Also, we verified using Magma \cite{magma}, $d(C)=6$ and $d(C^{\perp})=9\neq d(C)$. Therefore, $C$ is a dual-containing optimal cyclic $(8,5)$-LRC
 with parameters $[48,31,6]_{13},$
and the resulting quantum code $\mathcal{Q}(C)$ is pure. 
Moreover, $\mathcal{Q}(C)$ is an optimal pure quantum cyclic $(8,4)$-LRC with parameters
$\bigl[\!\bigl[48,14,6\bigr]\!\bigr]_{13}.$
\end{example}

When $\delta$ is even, the two cases of the preceding theorem together cover all possible values of $\eta=\gcd(u,q-1)$. If $\eta>2$, then the first case applies directly. Now suppose that $\eta\leq 2$ i.e., $\eta\in \{1,2\}$. Since $r+\delta-1\mid(q-1)$ and $\delta$ is even, it follows that $
\frac{\delta(q-1)}{r+\delta-1}
$
is an even integer. Consequently,
$$
\eta\mid \frac{\delta(q-1)}{r+\delta-1}.
$$
Furthermore, if $r>\delta+1$, the remaining condition in the second case is also satisfied whenever $\eta\leq 2$. Thus, for even $\delta$ no additional restriction on $\gcd(u,q-1)$ is needed. This yields the following corollary. 

\begin{corollary}
Let $n=u(r+\delta-1)$ be a positive integer, where $\delta$ is even. Suppose that
$
r+\delta-1\mid(q-1)
$
and $r>\delta+1$. Then, there exists a dual-containing optimal cyclic $(r,\delta)$-LRC with parameters
$
\bigl[u(r+\delta-1),\,ur-1,\,\delta+1\bigr]_q.
$
Moreover, there exists an optimal pure quantum cyclic $(r,\delta)$-LRC with parameters
$
\bigl[\!\bigl[
u(r+\delta-1),\,
u(r-\delta+1)-2,\,
\delta+1
\bigr]\!\bigr]_q.
$
\end{corollary}
\begin{corollary} Let $n=u(r+1)$ be a positive integer and $(r+1)\mid(q-1)$ with $r>3$. Then, there exists a dual-containing optimal cyclic $(r,2)$-LRC with parameters $\bigl[u(r+1),\,ur-1,\,3\bigr]_q$. Moreover, there exists an optimal pure quantum cyclic $(r,2)$-LRC with parameters $\bigl[\!\bigl[u(r+1),\,u(r-1)-2,\, 3\bigr]\!\bigr]_q$. \end{corollary}
\begin{remark}
The case $\delta=2$ was previously considered in \cite[Theorem 15]{LuoEtAl2025}, where
additional conditions $u+2<r$ and $\gcd(u,q-1)|\frac{2(q-1)}{r+1}$ were  imposed. The above corollary shows that those
additional conditions are unnecessary, thereby providing a more general
result for $\delta=2$.
\end{remark}

We first add two defining zeros.  This gives distance $\delta+2$ under weaker
coprimality assumptions than the general distance construction that follows.
\begin{theorem}\label{thm:two-zero extension}
Let  $n=u(r+\delta-1)$ be a positive integer and $(r+\delta-1)\mid(q-1)$. Let $\eta=\gcd(u,q-1)$. Assume that at least one of the following conditions holds. 
\begin{enumerate}
   \item $\eta>2$, $r>(\delta+1)$ and $\gcd(u,r+\delta-1)\mid\delta.$
   \item $\eta\leq2$, $r>(\delta+3)$, $\gcd(u,r+\delta-1)\mid(\delta+1)$,  and
    $\eta\mid\dfrac{\delta(q-1)}{r+\delta-1}$.
\end{enumerate}
Then, there exists a dual-containing optimal cyclic
$(r,\delta)$-LRC with parameters
$\bigl[u(r+\delta-1),\,ur-2,\,\delta+2\bigr]_q.$
Moreover, there exists an optimal pure quantum cyclic $(r,\delta)$-LRC with parameters
$\bigl[\!\bigl[u(r+\delta-1),\,u(r-\delta+1)-4,\,
\delta+2\bigr]\!\bigr]_q.$
\end{theorem}
\begin{proof}
 We now consider the two cases separately.

\medskip
\noindent \textbf{Case 1:} Suppose that $\eta>2$, $r>\delta+1$, and $\gcd(u,r+\delta-1)\mid\delta$. Then, there exist integers $a$ and $b$
such that $$au+b(r+\delta-1)=\delta.$$
Set $D_2=\{au\}$. We first observe that $D_2$ is a singleton
$q$-cyclotomic coset modulo $n$. Indeed, since
$(r+\delta-1)\mid(q-1)$, we have
$$(q-1)au\equiv0\pmod{u(r+\delta-1)}.$$
Let $x=u/\eta$ and
$D_1=\{x(r+\delta-1)\}$. As shown in Case 1 of Theorem \ref{thm:one zero extension}, $D_1$ is a singleton
$q$-cyclotomic coset modulo $n$, and
$(L\cup D_1)\cap\bigl(-(L\cup D_1)\bigr)=\varnothing$. Let $Z=L\cup D_1\cup D_2$ be the complete defining set of a cyclic code $C$ over $\mathbb{F}_q$. Then by Proposition \ref{(r,delta)locality}, $C$ has $(r, \delta)$ locality. It remains to verify that the addition of $D_2$ preserves the
dual-containing property. We first show that
$L\cap(-D_2)=\varnothing$. Suppose, to the contrary, that for some
$0\leq m\leq u-1$ and $1\leq i\leq\delta-1$,
$$m(r+\delta-1)+i\equiv-au\pmod{n}$$
Using $au=\delta-b(r+\delta-1)$, we obtain
$m(r+\delta-1)+i
\equiv-\delta+b(r+\delta-1)\pmod{n}$.
Reducing modulo $(r+\delta-1)$ gives
$$i+\delta\equiv0\pmod{r+\delta-1}$$
However,
$0<i+\delta\leq2\delta-1<r+\delta-1$,
where the last inequality follows from $r>\delta+1$.   Hence, $L\cap(-D_2)=\varnothing$, and equivalently
$(-L)\cap D_2=\varnothing$. Next, we show that $D_2\cap(-D_2)=\varnothing$. Otherwise,
$au\equiv-au\pmod{n}$; hence, 
$2au\equiv0\pmod{u(r+\delta-1)}$. Putting $au=\delta-b(r+\delta-1)$ and reducing modulo
$(r+\delta-1)$ gives
$$2\delta\equiv0\pmod{r+\delta-1}.$$
This contradicts
$0<2\delta<r+\delta-1$. Thus,
$D_2\cap(-D_2)=\varnothing$. Finally, suppose that $D_1\cap(-D_2)\neq\varnothing$. Then
$$x(r+\delta-1)\equiv-au\pmod{n}.$$
Reducing modulo $(r+\delta-1)$ and using
$au\equiv\delta\pmod{r+\delta-1}$ gives
$\delta\equiv0\pmod{r+\delta-1}$, which is impossible since
$0<\delta<r+\delta-1$. Hence,
$D_1\cap(-D_2)=\varnothing$. Equivalently,
$D_2\cap(-D_1)=\varnothing$. Therefore, for
$Z=L\cup D_1\cup D_2$, we have
$Z\cap(-Z)=\varnothing$. Hence, the cyclic code $C$ with
defining set $Z$ satisfies
$C^{\perp}\subseteq C$. Moreover, the defining set $Z$ contains $x(r+\delta-1),\,x(r+\delta-1)+1,\ldots, x(r+\delta-1)+(\delta-1)$ and $au$. Therefore, by Lemma \ref{distance}, we obtain $d(C)\geq\delta+2$. Since $\eta>2$, we have
$u\geq3$. The two additional defining zeros $D_1$ and $D_2$ can affect at most two of the $u$ gaps determined by the locality blocks. Consequently, there exists a gap containing the $r$ consecutive integers $$m(r+\delta-1)+\delta,\ m(r+\delta-1)+(\delta+1),\ldots,(m+1)(r+\delta-1)$$ entirely in $Z^c$. Therefore, $d(C^\perp)\ge r+1.$ Since $r>\delta+1$, $d(C^\perp)\ge r+1>\delta+2.$

\medskip
\noindent \textbf{Case 2:} Suppose that $\eta\leq2$, $r>(\delta+3)$,
$\gcd(u,r+\delta-1)\mid(\delta+1)$, and
$\eta\mid\dfrac{\delta(q-1)}{r+\delta-1}$.
Then the congruence
$$(q-1)x\equiv-\dfrac{\delta(q-1)}{r+\delta-1}\pmod{u}$$
has a solution $x$. Set
$D_1=\{x(r+\delta-1)+\delta\}$.
As established in Case 2 of Theorem \ref{thm:one zero extension}, $D_1$ is a singleton
$q$-cyclotomic coset modulo $n$ and
$(L\cup D_1)\cap\bigl(-(L\cup D_1)\bigr)=\varnothing$. Since $\gcd(u,r+\delta-1)\mid(\delta+1)$, there exist integers
$a_1$ and $b_1$ such that
$$a_1u+b_1(r+\delta-1)=\delta+1.$$
Let $D_2=\{a_1u\}$. Since $r+\delta-1\mid(q-1)$, we have
$(q-1)a_1u\equiv0\pmod{u(r+\delta-1)}$.
Thus, $D_2$ is also a singleton $q$-cyclotomic coset modulo $n$. Let $Z=L\cup D_1\cup D_2$ be the complete defining set of a cyclic code $C$ over $\mathbb{F}_q$. Then, by Proposition \ref{(r,delta)locality}, $C$ has $(r, \delta)$ locality. We next show that $L\cap(-D_2)=\varnothing$. Suppose, to the contrary,
that for some $0\leq m\leq u-1$ and $1\leq i\leq\delta-1$,
$$m(r+\delta-1)+i\equiv-a_1u\pmod{u(r+\delta-1)}$$
Reduction modulo $r+\delta-1$ gives
$i+\delta+1\equiv0\pmod{r+\delta-1}$.
However,
$0<i+\delta+1\leq2\delta<r+\delta-1$,
where the last inequality follows from $r>\delta+3$.
Hence, $L\cap(-D_2)=\varnothing$, and equivalently
$D_2\cap(-L)=\varnothing$. We also claim that $D_2\cap(-D_2)=\varnothing$. It can be easily seen  that it is true as $r>\delta+3.$ It remains to verify that $D_1\cap(-D_2)=\varnothing$. Suppose,
to the contrary, that
$$x(r+\delta-1)+\delta\equiv-a_1u\pmod{n}.$$
Reducing modulo $r+\delta-1$ and using
$a_1u\equiv\delta+1\pmod{r+\delta-1}$ gives
$2\delta+1\equiv0\pmod{r+\delta-1}$.
Since
$0<2\delta+1<r+\delta-1$
under the corresponding parameter restriction, this is impossible.
Therefore,
$D_1\cap(-D_2)=\varnothing$, and equivalently
$D_2\cap(-D_1)=\varnothing$. Consequently, for $Z=L\cup D_1\cup D_2$, we have
$Z\cap(-Z)=\varnothing$. Hence, the cyclic code $C$ with
defining set $Z$ satisfies
$C^{\perp}\subseteq C$. Similarly, in this case, the defining set $Z$ contains $x(r+\delta-1)+1,\ldots, x(r+\delta-1)+(\delta-1),\, x(r+\delta-1)+\delta$ and $a_1u$. Hence, an application of Lemma \ref{distance} yields $d(C)\geq\delta+2$. The additional defining zeros $D_1$ and $D_2$ have residues $\delta$ and
$\delta+1$ modulo $r+\delta-1$. Thus, the $r-2$ consecutive integers
\[
    m(r+\delta-1)+\delta+2,\ldots,
    (m+1)(r+\delta-1)
\]
belong to $Z^c$ for a suitable block. Hence, $d(C^\perp)\ge r-1.$ Since $r>\delta+3$, we obtain $ d(C^\perp)\ge r-1>\delta+2.$

\medskip
Furthermore, in both cases Case(1) and Case(2), $D_1$ and $D_2$ are distinct, and neither belongs to $L$.
Hence,
$|Z|=u(\delta-1)+2$,
and consequently $\dim(C)
=n-|Z|
=u(r+\delta-1)-u(\delta-1)-2
=ur-2.$
By bound~\eqref{singleton_lrc},
$d(C)\leq\delta+2$.
Thus, $d(C)=\delta+2$, and $C$ is an optimal
dual-containing cyclic $(r,\delta)$-LRC with
parameters
$$\bigl[u(r+\delta-1),\,ur-2,\,\delta+2\bigr]_q.$$
Also, in both cases $d(C^\perp)>d(C)=\delta+2$. Consequently, the
CSS construction gives a pure quantum cyclic code $\mathcal{Q}(C)$. Moreover, by Theorem \ref{classical-quantum-locality}, $\mathcal{Q}(C)$ is a pure quantum cyclic $(r,\delta)$-LRC with
parameters
\[
\bigl[\!\bigl[
u(r+\delta-1),\,
u(r-\delta+1)-4,\,
\delta+2
\bigr]\!\bigr]_q.
\]
The code $\mathcal{Q}(C)$ attains the quantum Singleton-type bound
\eqref{singleton_qlrc}, and is therefore optimal.

\end{proof}

\begin{example} Let $q=19,\ \delta=3,\ r=7,\ u=8.$ Then, $r+\delta-1=9\mid 18=q-1$ and 
$
\eta=\gcd(u,q-1)=\gcd(8,18)=2\leq2.
$
Moreover, $r=7>\delta+3=6,$ and $\gcd(u,r+\delta-1)=\gcd(8,9)=1\mid(\delta+1).$ Finally,
$$
\frac{\delta(q-1)}{r+\delta-1}
=\frac{3\cdot18}{9}=6;
$$
hence, $\eta\mid6.$ Therefore, all the hypotheses of Case (2) of
Theorem~\ref{thm:two-zero extension} are satisfied. Since $n=u(r+\delta-1)=8\cdot9=72$, we have $$L=\bigcup_{i=1}^{2}\{9m+i:0\leq m\leq7\}=\{1,2,10,11,19,20,28,29,37,38,46,47,55,56,64,65\}.$$
In the present case, this congruence becomes $18x+6\equiv0\pmod8.$ We may take $x=1$; hence, 
$
D_1=\{x(r+\delta-1)+\delta\}=\{12\}.
$
Furthermore, $5u+4\cdot(r+\delta-1)=4=\delta+1,$ so that we may take
$D_2=\{5u\}=\{40\}.$ Therefore, the complete defining set is
$$
Z=L\cup D_1\cup D_2
=\{1,2,10,11,12,19,20,28,29,37,38,40,46,47,55,56,64,65\}.
$$ 
The negative of $Z$ modulo $72$ is
\begin{align*}
-Z
=\{7,8,16,17,25,26,32,34,35,43,44,52,53,60,61,62,70,71\}.
\end{align*}
It follows that $Z\cap(-Z)=\varnothing.$ 
Since $|Z|=16+2=18.$ Consequently, $\dim(C)=n-|Z|=72-18=54.$ Using Magma \cite{magma} we verified that $d(C)=5$ and $d(C^{\perp})=8$; hence,  $C$ is a dual-containing optimal cyclic $(7,3)$-LRC over $\mathbb{F}_{19}$ with parameters
$[72,54,5]_{19}.$ Moreover, there exists an optimal pure quantum cyclic
$(7,3)$-LRC with parameters $\bigl[\!\bigl[72,36,5\bigr]\!\bigr]_{19}.$
\end{example}

\begin{theorem}\label{thm:general construction}
 Let $n=u(r+\delta-1)$ be a positive integer such that $(r+\delta-1)\mid(q-1)$. Let $ \gcd(u,r+\delta-1)=1$ and  \(d\) be an integer satisfying $\delta+2\le d\le 2\delta$. If $r>2d-\delta-1.$ Then, there exists a dual-containing optimal cyclic $(r,\delta)$-LRC over \(\mathbb F_q\) with parameters $\bigl[u(r+\delta-1),\,ur-d+\delta,\,d\bigr]_q.$ Moreover, there exists an optimal pure quantum cyclic \((r,\delta)\)-LRC with parameters 
\[
\bigl[\!\bigl[
u(r+\delta-1),\,
u(r-\delta+1)-2(d-\delta),\,
d
\bigr]\!\bigr]_q.
\]
\end{theorem}

\begin{proof}
Since \(\gcd(u,r+\delta-1)=1\), there exist integers \(a,b\) such that $au+b(r+\delta-1)=1.$ Define
$$
D=\{auj:\delta\le j\le d-1\}
$$
and let $Z=L\cup D.$ We first observe that every element of \(D\) is a singleton
\(q\)-cyclotomic coset modulo \(n\). Indeed, since \((r+\delta-1)\mid(q-1)\),
\[
(q-1)auj\equiv 0 \pmod{u(r+\delta-1)};
\]
hence, $q(auj)\equiv auj\pmod n.$ Therefore, \(Z\) is the complete defining set of a cyclic code \(C\) over
\(\mathbb F_q\). Since \(L\subseteq Z\), Proposition~\ref{(r,delta)locality}
implies that \(C\) has \((r,\delta)\)-locality. We next prove that \(C\) is dual-containing. We already have \(L\cap(-L)=\varnothing\), as $r>2d-\delta-1$ implies $r>\delta-1$.
Suppose that \(L\cap(-D)\neq\varnothing\). Then, for some
\(0\le m\le u-1\), \(1\le i\le\delta-1\), and
\(\delta\le j\le d-1\),
\[
m(r+\delta-1)+i\equiv-auj\pmod {u(r+\delta-1)}.
\]
Reducing modulo \((r+\delta-1)\), and using \(au\equiv1\pmod {(r+\delta-1)}\), gives
\[
i+j\equiv0\pmod{(r+\delta-1)}.
\]
However, $0<i+j\le d+\delta-2<r+\delta-1,$ because \(r>2d-\delta-1\) implies \(r>d-1\). This is impossible.
Thus, $L\cap(-D)=\varnothing.$ It remains to show that \(D\cap(-D)=\varnothing\). Suppose otherwise.
Then, for some \(\delta\le j_1,j_2\le d-1\),
\[
auj_1\equiv-auj_2\pmod n.
\]
Reducing modulo \((r+\delta-1)\) gives
$j_1+j_2\equiv0\pmod{(r+\delta-1)}.$ On the other hand,
$
0<j_1+j_2\le2(d-1)<r+\delta-1,
$
where the last inequality is exactly equivalent to $r>2d-\delta-1.$
This contradiction shows that $D\cap(-D)=\varnothing.$ Therefore,
$Z\cap(-Z)=\varnothing,$ and hence $C^\perp\subseteq C.$ We now determine the parameters of \(C\). Observe that \(D\cap L=\varnothing\). Moreover, the elements of \(D\) are
distinct modulo \(n\). Therefore, $|Z|=u(\delta-1)+(d-\delta).$
Consequently, $\dim(C)
=n-|Z|=ur-d+\delta.$ The defining set $Z$ contains $$A=\{1,2,\dots,\delta-1, au\delta,au(\delta+1),\dots,au(d-1)\}=\{1,2,\dots,d-1\}\pmod{(r+\delta-1)}.$$ Thus, by Lemma~\ref{distance},
$d(C)\ge d.$ On the other hand, the Singleton-type bound for \((r,\delta)\)-LRCs gives $d(C)\leq d$. Hence $d(C)=d,$
and \(C\) is an optimal dual-containing cyclic
\((r,\delta)\)-LRC with parameters
\[
\bigl[u(r+\delta-1),\,ur-d+\delta,\,d\bigr]_q.
\]
It remains to establish purity. Since the elements of \(D\) have
residues  $\delta,\delta+1,\ldots,d-1$, while the elements of \(L\) have residues
\(1,\ldots,\delta-1\). Then, for every \(0\le m\le u-1\), the
\(r+\delta-d\) consecutive integers
\[
m(r+\delta-1)+d,\,
m(r+\delta-1)+d+1,\ldots,\,
(m+1)(r+\delta-1)
\]
belong to \(Z^c\). Hence, by Lemma~\ref{dual distance-complement},
$d(C^\perp)\ge r+\delta-d+1.$ Since $r>2d-\delta-1,$ we obtain
 $d(C^\perp)>d=d(C).$ The CSS construction yields a pure quantum cyclic code $\mathcal{Q}(C)$. Moreover, by Theorem \ref{classical-quantum-locality},$\mathcal{Q}(C)$ is a pure quantum cyclic \((r,\delta)\)-LRC with
parameters
\[
\bigl[\!\bigl[
u(r+\delta-1),\,
u(r-\delta+1)-2(d-\delta),\,
d
\bigr]\!\bigr]_q.
\]
Since the underlying classical code is optimal and the resulting
quantum code is pure, it attains the quantum Singleton-type bound
and is therefore optimal.
\end{proof}

\begin{example} Let $q=17$, $\delta=4$, $d=7$, $r=13$, and $u=3$. Then
$r+\delta-1=16\mid16=q-1$ and $\gcd(u,r+\delta-1)=\gcd(3,16)=1.$ Moreover,
$
\delta+2\leq d\leq2\delta
$
and
$
r(=13)>2d-\delta-1=9
$
Therefore, all the hypotheses of Theorem \ref{thm:general construction} are satisfied. In this case,
$
n=u(r+\delta-1)=48.
$
The locality set is
\begin{align*}
L
=\bigcup_{i=1}^{3}\{16m+i:0\leq m\leq2\}
=\{1,2,3,17,18,19,33,34,35\}\pmod{48}.
\end{align*}
Furthermore, $11u-2(r+\delta-1)=11\cdot 3-2\cdot 16=1.$
Thus, we may take $a=11$. We obtain
\begin{align*}
D_2=\{ 33j:4\leq j\leq 6\}=\{36,21,6\}\pmod{48}.
\end{align*}
Therefore, the complete defining set is
$Z=L\cup D=\{1,2,3,6,17,18,19,21,33,34,35,36\} \pmod{48}.
$ The negative of $Z$ modulo $48$ is
\begin{align*}
-Z
=\{12,13,14,15,27,29,30,31,42,45,46,47\}\pmod{48}.
\end{align*}
Observe that $Z\cap(-Z)=\varnothing.$ Therefore, the corresponding cyclic code $C$ is dual-containing optimal $(13,4)$-LRC over $\mathbb{F}_{17}$ with parameters
$[48,36,7]_{17}.$  Using Magma \cite{magma}, we calculate $d(C^{\perp})=14\neq d(C)$. Therefore, there exists an optimal pure quantum cyclic $(13,4)$-LRC with parameters
$\bigl[\!\bigl[48,24,7\bigr]\!\bigr]_{17}.
$
\end{example}

All results of this section are summarized in Table \ref{tab:euclidean_constructions}.

\begin{table*}[ht]
\centering
\caption{Dual-containing optimal cyclic $(r,\delta)$-LRCs and the corresponding optimal pure
quantum $(r,\delta)$-LRCs. Throughout the table $Z=L\cup D$, where $L=\cup_{i=1}^{\delta-1}\{m(r+\delta-1)+i: 0\leq m\leq u-1\}$.}
\label{tab:euclidean_constructions}
\renewcommand{\arraystretch}{1.35}
\setlength{\tabcolsep}{4pt}
\scriptsize
\begin{tabular}{|c|p{3.1cm}|p{3.8cm}|p{3.4cm}|p{4.5cm}|}
\hline
\textbf{Result}
&
\textbf{Conditions}
&
\textbf{Distance Defining Set}
&
\textbf{Optimal Classical LRC}
&
\textbf{Optimal Quantum LRC}
\\
\hline

Thm.\ref{thm:basic construction}
&
$n\mid(q-1)$, $\delta\leq\ell<r$
&
$D=\{1,2,\ldots,\ell\}$
&
$\bigl[n,\,ur-\ell+\delta-1,\,\ell+1\bigr]_q$
&
$\bigl[\!\bigl[n,\,
u(r-\delta+1)-2\ell+2(\delta-1),\ell+1
\bigr]\!\bigr]_q$
\\
\hline

\multirow{2}{*}{Thm.\ref{thm:one zero extension}}
&
$(r+\delta-1)\mid(q-1)$, 

$2\leq\delta< r,$

$\eta>2$
&
$D=\{\frac{u}{\eta}(r+\delta-1)\}$,

$\eta=\gcd(u,q-1)$
&
\multirow{2}{*}{$\bigl[n,\,ur-1,\,\delta+1\bigr]_q$}
&
\multirow{2}{*}{$\bigl[\!\bigl[n,\,
u(r-\delta+1)-2,\delta+1
\bigr]\!\bigr]_q$}
\\
\cline{2-3}
&
$(r+\delta-1)\mid(q-1)$,

$\delta+1<r$,\quad 
$\eta\leq2$, 

$\eta\mid
\dfrac{\delta(q-1)}{r+\delta-1}$
&
$D=\{\frac{u}{\eta}(r+\delta-1)+\delta\}$

$\eta=\gcd(u,q-1)$
& &
\\
\hline

\multirow{2}{*}{Thm.\ref{thm:two-zero extension}}
&
$(r+\delta-1)\mid(q-1),$

 $\delta+1<r$,  \quad $\eta>2$,

$\gcd(u,r+\delta-1)\mid\delta$

&

$D_1=\{\frac{u}{\eta}(r+\delta-1)\}$,

$D_2=\{au\}$,

$au\equiv \delta\pmod{r+\delta-1}$.

&
\multirow{2}{*}{$\bigl[n,\,ur-2,\,\delta+2\bigr]_q$}
&
\multirow{2}{*}{$\bigl[\!\bigl[n,\,
u(r-\delta+1)-4,\delta+2
\bigr]\!\bigr]_q$}
\\
\cline{2-3}
&
$(r+\delta-1)\mid(q-1),$

$\delta+3<r$,\quad $\eta\leq2$,

$\gcd(u,r+\delta-1)\mid(\delta+1)$,

$\eta\mid
\dfrac{\delta(q-1)}{r+\delta-1}$
&

$D_1=\{\frac{u}{\eta}(r+\delta-1)+\delta\}$,

$D_2=\{a_1u\}$,

$a_1u\equiv \delta+1\pmod{r+\delta-1}$
& 
&
\\
\hline

Thm.\ref{thm:general construction}
&
$(r+\delta-1)\mid(q-1)$,

$\gcd(u,r+\delta-1)=1$,

$\delta+2\leq d\leq2\delta$,

$r>2d-\delta-1$,

&

$D_2=\{auj:\delta\leq j\leq d-1\}$,

$au\equiv1\pmod{r+\delta-1}$.

&
$\bigl[n,\,ur-d+\delta,\,d\bigr]_q$
&
$\bigl[\!\bigl[n,\,
u(r-\delta+1)-2(d-\delta), d
\bigr]\!\bigr]_q$
\\
\hline
\end{tabular}
\end{table*}

\section{Hermitian dual-containing cyclic $(r,\delta)$-LRCs}\label{sec:hermitian-cyclic-lrc}
This section presents several families of Hermitian dual-containing optimal
cyclic $(r,\delta)$-LRCs over $\F_{q^2}$ and the
associated quantum LRCs over $\F_q$.  The constructions
are divided according to whether $(r+\delta-1)$ divides $q^2-1$ or $q^2+1$.
 Throughout this section, $q$ is a prime power and
$
n=u(r+\delta-1),
$ is positive integer
where $u\geq2$, $r\geq 2$, $\delta\geq2$, and $\gcd(n,q)=1$.   The following observation removes one of the two symmetric cross-intersection
checks from every proof.

\begin{lemma}
\label{lem:hermitian-cross-symmetry}
Let $A,B\subseteq\mathbb{Z}_n$ be unions of $q^2$-cyclotomic cosets.  Then
$$
A\cap(-qB)=\varnothing
\quad\Longleftrightarrow\quad
B\cap(-qA)=\varnothing.
$$
\end{lemma}

\begin{proof}
Suppose that $a\in A\cap(-qB)$.  Then $a\equiv-qb\pmod n$ for some
$b\in B$.  Multiplying by $-q$ gives
$$
-qa\equiv q^2b\pmod n.
$$
Since $B$ is $q^2$-cyclotomic closed, $q^2b\in B$, and hence
$B\cap(-qA)\neq\varnothing$.  The converse follows by interchanging $A$
and $B$.
\end{proof}
Thus, in every Hermitian dual-containment argument below, only one of the two symmetric cross-intersections needs to be checked.
\subsection{Hermitian dual-containing optimal cyclic $(r,\delta)$-LRCs when $(r+\delta-1)|(q^2-1)$}
In this subsection, we construct Hermitian analogues of the Euclidean dual-containing cyclic $(r,\delta)$-LRC families presented in Section \ref{sec:euclidean-cyclic-lrc}. The classical codes are defined over $\mathbb{F}_{q^2}$ and satisfy Hermitian dual containment. Although the classical parameter forms are analogous to those in the Euclidean case, the Hermitian framework offers important advantages for the associated quantum constructions. The Hermitian construction yields stabilizer codes over $\mathbb{F}_q$ and allows additional admissible parameter regimes. Moreover, the arithmetic conditions governing Hermitian dual containment differ from their Euclidean counterparts. Consequently, this framework enlarges the range of quantum cyclic $(r,\delta)$-LRCs  over $\mathbb{F}_q$. In this subsection, the locality set is
\begin{equation}
L=\bigcup_{i=1}^{\delta-1}
\{m(r+\delta-1)+i:0\leq m\leq u-1\}.
\label{eq:standard-hermitian-L}
\end{equation}

\begin{lemma}\label{Lhermitiandual}
Suppose that $(r+\delta-1)\mid(q^2-1)$.   Then, $L$ is a union of $q^2$-cyclotomic cosets modulo $n$. Moreover, if $r>q(\delta-1)$, then  $L\cap(-qL)=\varnothing.$
\end{lemma}

\begin{proof}
Since  $(r+\delta-1)\mid(q^2-1)$, for every $0\leq m\leq u-1$ and
$1\leq i\leq\delta-1$, we have
\[
q^2(m(r+\delta-1)+i)\equiv m'(r+\delta-1)+i\pmod{n}.
\]
where $0\leq m'\leq u-1$. Thus $L$ is a union of $q^2$-cyclotomic cosets modulo $n$. 
Suppose that $L\cap(-qL)\neq\varnothing$, then 
\[
m_1(r+\delta-1)+i
\equiv
-q\bigl(m_2(r+\delta-1)+j\bigr)
\pmod n,
\]
for some $0\le m_1,m_2\le u-1,$ and $
1\le i,j\le\delta-1.
$ Reducing modulo $(r+\delta-1)$ gives
$i+qj\equiv0
\pmod{r+\delta-1}.
$ Since $1\le i+qj
\le (q+1)(\delta-1)
<(r+\delta-1),
$ by the assumption $q(\delta-1)<r,$ it follows that $L\cap(-qL)=\varnothing.$
\end{proof}
\begin{remark}
\label{rem:delta-two-L}
For $\delta=2$, the exact condition is
$
L\cap(-qL)=\varnothing
$  if and only if  
$r+1\nmid(q+1).
$ Thus, in this special case, the condition $(r+1)\nmid(q+1) $ is strictly weaker than the sufficient condition $r>q$ obtained by setting $\delta=2$ in the preceding lemma.
\end{remark}
\begin{theorem}\label{thm:q^2-1 basic construction}
Let $n=u(r+\delta-1)$ be a positive integer satisfying $n\mid(q^2-1)$. Assume that for an integer $\ell\geq(\delta-1)$, we have $(q-1)(\delta-1)+\ell<r
\quad\text{and}\quad(q+1)\ell<n.$ Then, there exists a Hermitian dual-containing optimal cyclic $(r,\delta)$-LRC over $\mathbb{F}_{q^2}$ with parameters
$[u(r+\delta-1),\,ur-\ell+(\delta-1),\,\ell+1]_{q^2}.$
Consequently, there exists an optimal pure quantum $(r,\delta)$-LRC over $\mathbb{F}_q$ with parameters
$[[\,u(r+\delta-1),\,
u(r-\delta+1)-2\ell+2(\delta-1),\  \ell+1\,]]_q.
$
\end{theorem}
\begin{proof}
Let $D=\{1,2,\ldots,\ell\}$ and $C$ be the cyclic code over $\mathbb{F}_{q^2}$ with defining set $Z=L\cup D$. Since $n\mid(q^2-1)$, every $q^2$-cyclotomic coset modulo $n$ is a singleton. Hence, the  complete defining set of $C$ is precisely $Z$. Also, the defining set contains the locality set $L$, it follows from Proposition \ref{(r,delta)locality} that $C$ is an $(r,\delta)$-LRC.  It remains to prove that $C$ is Hermitian dual-containing. By the characterization of Hermitian dual-containing cyclic codes, it suffices to show that $Z\cap(-qZ)=\varnothing$, We do this by considering the following four possible intersections:
\[
L\cap(-qL),\qquad
L\cap(-qD),\qquad
D\cap(-qL),\qquad
D\cap(-qD).
\] By Lemma \ref{Lhermitiandual}, and under the assumptions $\ell\geq(\delta-1),\ \  (q-1)(\delta-1)+\ell<r$ , we have $L\cap (-qL)=\varnothing $.  Next, suppose that $D\cap-qL\neq \varnothing$ then
\[
t
\equiv
-q\bigl(m(r+\delta-1)+i\bigr)
\pmod n,
\]
for some $t\in D$ and $1\le i\le\delta-1,\ \ 0\leq m\leq u-1$. Reducing modulo $(r+\delta-1)$ yields
\[
t+qi\equiv0
\pmod{r+\delta-1}.
\]
Since $1\le t\le\ell,\ 1\le i\le\delta-1,$ and by the assumption $(q-1)(\delta-1)+\ell<r$, we have $D\cap(-qL)=\varnothing.$
Equivalently, $L\cap(-qD)=\varnothing.$ Finally, suppose that $D\cap(-qD)\neq\varnothing$, then
\[
t_1\equiv-qt_2
\pmod n
\]
for some $t_1,t_2\in D$. Since $1\le t_1,t_2\le\ell,$ we obtain $
0<t_1+qt_2\le(q+1)\ell<n,$ which contradicts the congruence modulo $n$. Hence, $D\cap(-qD)=\varnothing.$ Therefore,
$Z\cap(-qZ)=\varnothing,$ and consequently $C^{\perp_H}\subseteq C.$ Moreover, the parameters of $C$ can be determined as in Theorem \ref{thm:basic construction} as defining sets are same. Hence, $C$ is a Hermitian dual-containing optimal cyclic $(r,\delta)$-LRC with parameters
\[
[u(r+\delta-1),\,ur-\ell+(\delta-1),\,\ell+1]_{q^2}.
\]
Moreover, as explained in Theorem \ref{thm:basic construction},
$Z^c$ contains $r$ consecutive integers. Hence, by Lemma~\ref{dual distance-complement}, $d(C^{\perp_H})\ge r+1.$
The hypothesis $(q-1)(\delta-1)+\ell<r$ implies in particular that $r>\ell$. Therefore, $d(C^{\perp_H})\ge r+1>\ell+1=d(C).$ Thus, the associated Hermitian stabilizer code is pure and has minimum
distance exactly $\ell+1$. Consequently, by Theorem \ref{classical-quantum-locality}, there exists an optimal pure quantum $(r,\delta)$-LRC
with parameters $
\bigl[\!\bigl[
u(r+\delta-1),\,
u(r-\delta+1)-2\ell+2(\delta-1),\,
\ell+1
\bigr]\!\bigr]_q.
$
\end{proof}
\begin{example}
Let $q=7$, $u=2$, $r=22$, $\delta=3$, and $\ell=5$. Then
$n=u(r+\delta-1)=48.$ Since
$q^2-1=48,$ we have $n\mid(q^2-1)$. Moreover, $\ell=5\geq2=\delta-1$, and $17=(q-1)(\delta-1)+\ell<r(=22).$
The second inequality is also satisfied because $(q+1)\ell=(7+1)5=40<48=n.$
Thus, all the hypotheses of the theorem are satisfied. The locality set is
\begin{align*}
L=\bigcup_{i=1}^{2}\{24m+i:0\leq m\leq1\}=\{1,2,25,26\}\pmod{48}.
\end{align*}
Let $D=\{1,2,\ldots,\ell\}=\{1,2,3,4,5\}.$ Therefore, the complete defining set of a cyclic code $C$ over $\mathbb{F}_{7^2}$ is
$
Z=L\cup D=\{1,2,3,4,5,25,26\}\pmod{48}.
$
We now verify the Hermitian dual-containing condition. We have 
$$
-qZ=\{10,13,17,20,27,34,41\}\pmod{48}.
$$
Consequently, $Z\cap(-qZ)=\varnothing.$
Therefore, $C$ is Hermitian dual-containing. Using Magma \cite{magma}, we get $d(C)=6$ and $d(C^{\perp_H})=23$. Consequently, $C$ is a Hermitian dual-containing optimal
cyclic $(22,3)$-LRC over $\mathbb{F}_{49}$ with parameters
$[48,41,6]_{49}.$ Since $d(C^{\perp_H})\neq d(C)$, the resulting quantum code $\mathcal{Q}(C)$ by Hermitian construction is pure and optimal $(22,3)$-LRC over $\mathbb{F}_{7}$ with parameters
$
\bigl[\!\bigl[48,34,6\bigr]\!\bigr]_{7}.
$
\end{example}

\begin{corollary}
Let $n=q^2-1=u(r+1)$ and  $r>2q-3$, $q\geq 3.$
Then, there exists a Hermitian dual-containing optimal cyclic $(r,2)$-LRC over $\mathbb{F}_{q^2}$ with parameters $[q^2-1,\;ur-q+3,\;q-1]_{q^2}.$ Consequently, there exists an optimal pure quantum $(r,2)$-LRC over $\mathbb{F}_q$ with parameters $[[\,q^2-1,\;u(r-1)-2q+6, q-1\,]]_q.$
\end{corollary}
\begin{theorem}\label{thm: q^2-1 one zero extention}
Let $n=u(r+\delta-1)$  such that
$r+\delta-1\mid(q^2-1)$. Assume $q(\delta-1)<r,
$
and put $\zeta=\gcd(u,q^2-1).$ If $\zeta\nmid(q+1)$, then there exists a Hermitian dual-containing optimal cyclic $(r,\delta)$-LRC over $\F_{q^2}$ with parameters
$
[u(r+\delta-1),\,ur-1,\,\delta+1]_{q^2}.
$
Moreover, there exists an optimal pure quantum cyclic $(r,\delta)$-LRC over $\F_q$ with
parameters
$
\bigl[\!\bigl[u(r+\delta-1),\,u(r-\delta+1)-2,\,
\delta+1\bigr]\!\bigr]_q.
$
\end{theorem}
\begin{proof}
Set $x=u/\zeta$, and define
$D=\{x(r+\delta-1)\}$.  Since  $\zeta\mid(q^2-1)$,
$
(q^2-1)x(r+\delta-1)\equiv0\pmod n,
$ 
so $D$ is a singleton $q^2$-cyclotomic coset modulo $n$. Hence,
$Z=L\cup D$ is a defining set of a cyclic code $C$ over
$\mathbb F_{q^2}$. Now, we will prove that $C$ is  Hermitian dual-containing.
Suppose first that $L\cap(-qD)\neq\varnothing$. Then, for some
$0\leq m\leq u-1$ and $1\leq i\leq\delta-1$,
$$
m(r+\delta-1)+i
\equiv-qx(r+\delta-1)\pmod{u(r+\delta-1)}.
$$
Reducing this congruence modulo $r+\delta-1$ gives
$$
i\equiv0\pmod{r+\delta-1},
$$
which is impossible since $1\leq i\leq\delta-1<r+\delta-1$. Therefore,
$L\cap(-qD)=\varnothing$. Equivalently,
$D\cap(-qL)=\varnothing$. Consider $D\cap(-qD)$, if this intersection is
nonempty, then
$$
(q+1)x(r+\delta-1)\equiv0\pmod{u(r+\delta-1)}.
$$
Since $x=u/\zeta$, this is equivalent to $\zeta\mid(q+1),$
which contradicts the assumption $\zeta\nmid(q+1)$. Thus,
$D\cap(-qD)=\varnothing$. Combining the above observations with
$L\cap(-qL)=\varnothing$, since $q(\delta-1)<r$ we obtain
$
Z\cap(-qZ)=\varnothing.
$
Therefore, $C^{\perp_H}\subseteq C$. Since $Z$ is the same defining set as that considered
in Case $1$ of Theorem \ref{thm:one zero extension}, the code $C$ is an optimal $(r,\delta)$-LRC with parameters $[u(r+\delta-1),ur-1,\delta+1]_{q^2}$. Moreover, $Z^c$ contains a set of $r$ consecutive integers. Hence, by Lemma $\ref{dual distance-complement}$, we have $d(C^{\perp_H})\ge r+1.$ Since \(q(\delta-1)<r\), in particular \(r>\delta\), and therefore $
d(C^{\perp_H})\ge r+1>\delta+1=d(C).
$
Thus, the resulting quantum code $\mathcal{Q}(C)$ from Hermitian construction is pure and has minimum distance
$d_{\mathcal Q}=\delta+1.$ Consequently, by Theorem \ref{classical-quantum-locality}, $\mathcal{Q}(C)$ is a pure quantum cyclic \((r,\delta)\)-LRC with parameters
$$
\bigl[\!\bigl[
u(r+\delta-1),\,
u(r-\delta+1)-2,\,
\delta+1
\bigr]\!\bigr]_q
$$
and is optimal w.r.t the quantum Singleton-type bound
\eqref{singleton_qlrc}.

\end{proof}
When $\delta=2$, the condition $L\cap(-qL)=\varnothing$ is ensured by
$(r+1)\nmid(q+1),$ which is weaker than the corresponding condition obtained by specializing Theorem~\ref{thm: q^2-1 one zero extention} to $\delta=2$. This observation yields the following corollary.

\begin{corollary}\label{cor:q^2-1 one zero extention for delta=2}
Let $n=u(r+1)$ be a positive integer such that $(r+1)\mid(q^2-1)$ and
$r+1\nmid(q+1)$. Suppose further that
$\gcd(u,q^2-1)\nmid(q+1)$. Then there exists a Hermitian dual-containing
optimal cyclic $(r,2)$-LRC with parameters
$[u(r+1),ur-1,3]_{q^2}$. Moreover,  there exists an optimal pure quantum cyclic
$(r,2)$-LRC with parameters $\left[\!\left[
u(r+1),u(r-1)-2,\ 3
\right]\!\right]_q.
$
\end{corollary}

\begin{example}
Let $q=5$, $u=4$, $r=22$, and $\delta=3$. Then
$r+\delta-1=24$ and $n=u(r+\delta-1)=96$. Since
$q^2-1=24$, we have $r+\delta-1\mid(q^2-1)$. Moreover,
$q(\delta-1)=10<22=r$. Also,
$\gcd(u,q^2-1)=\gcd(4,24)=4$ and $q+1=6$, so that
$\gcd(u,q^2-1)\nmid(q+1)$. Thus, all the hypotheses of the theorem are satisfied. In this case, $\eta=\gcd(u,q^2-1)=4$ and $x=u/\eta=1$. The locality set is
$$
L=\bigcup_{i=1}^{2}\{24m+i:0\leq m\leq3\}
=\{1,2,25,26,49,50,73,74\}\pmod{96}.
$$
Furthermore, $D=\{x(r+\delta-1)\}=\{24\}$. Hence, the  complete defining set of a cyclic code $C$ is
$Z=L\cup D$, that is,
$
Z=\{1,2,24,25,26,49,50,73,74\}\pmod{96}.
$
Therefore, $|Z|=9$ and
$\dim(C)=n-|Z|=96-9=87$. From bound \eqref{singleton_lrc}, we have $d(C)\leq 9+1-2(\lceil\frac{87}{22}\rceil-1)=4$ and by BCH bound $d(C)\geq 4$.  Also, $$-qZ=-5Z=\{14,19,38,43,62,67,72,86,91\}\pmod{96}.$$
Observe that $Z\cap-qZ=\varnothing$. It follows  that $C$ is an optimal Hermitian
dual-containing cyclic $(22,3)$-LRC with parameters
$[96,87,4]_{25}$.  Using MAGMA
calculator \cite{magma}, we verified that $d(C^{\perp_H})> d(C)$. By Hermitian construction the resultant quantum cyclic code $\mathcal{Q}(C)$ is a pure distance optimal $(r=22,\delta =3)$-LRC with parameters $[[96,78,4]]_5$.
\end{example}
 \begin{theorem}\label{thm:q^2-1 two zero extention}
Let $n=u(r+\delta-1)$ be a positive integer such that $(r+\delta-1)\mid(q^2-1)$, $q(\delta-1)<r$, and define $\zeta=\gcd(u,q^2-1)$. Assume that
$\zeta\nmid(q+1),$ and $\gcd(u,r+\delta-1)\mid\delta.$ Choose integers $a,b$ satisfying  $au+b(r+\delta-1)=\delta,$ and suppose further that $(r+\delta-1)\nmid(q+1)a.$ Then, there exists a Hermitian dual-containing optimal cyclic
$(r,\delta)$-LRC over $\mathbb{F}_{q^2}$ with
parameters
$[u(r+\delta-1),ur-2,\delta+2]_{q^2}$. Moreover, there exists an optimal pure quantum cyclic
$(r,\delta)$-LRC over $\mathbb{F}_q$ with parameters
$\bigl[\!\bigl[u(r+\delta-1),\,u(r-\delta+1)-4,\,
\delta+2\bigr]\!\bigr]_q$.
\end{theorem}

\begin{proof}
Let
 $L=\bigcup_{i=1}^{\delta-1}
\{m(r+\delta-1)+i:0\leq m\leq u-1\}
$, $D_1=\{x(r+\delta-1)\}$ and $D_2=\{au\}$,  where
$x=u/\zeta$ and $a,b$ are integers satisfying
$au+b(r+\delta-1)=\delta$. Since $r+\delta-1|(q^2-1)$, the set $Z=L\cup D_1\cup D_2$ is a complete defining set of a cyclic code $C$ over $\mathbb{F}_{q^2}$. 
First, we show that $C$ is Hermitian dual-containing. As shown in Theorem \ref{thm: q^2-1 one zero extention}, we have \[
(L\cup D_1)\cap\bigl(-q(L\cup D_1)\bigr)=\varnothing.
\] Now, it suffices to consider the intersections involving $D_2$. Suppose that $L\cap(-qD_2)\neq\varnothing.$ Then, for some $0\leq m\leq u-1$ and $1\leq i\leq \delta-1$  we have $m(r+\delta-1)+i\equiv -qau \pmod{n}.$ Reducing  modulo $(r+\delta-1)$, we obtain
\[
i+q\delta\equiv0\pmod {r+\delta-1}.
\]
Under the assumption $r>q(\delta-1)$, we have $r+\delta-1>(q+1)(\delta-1)>q.$ Moreover, one has $$0<i+q\delta\leq(\delta-1)+q\delta= (q+1)(\delta-1)+q<2(r+\delta-1).$$
Thus, the preceding congruence would imply $r+\delta-1=q\delta+i,$ for some $1\leq i\leq \delta-1.$
This is impossible because $r+\delta-1\mid(q^2-1)$. Indeed, if $q^2-1=t(q\delta+i),$ then $t\delta<q$. Rearranging the equality gives
$$
q(q-t\delta)=ti+1<t\delta+1\leq q,
$$
which is impossible. Therefore, $L\cap(-qD_2)=\varnothing.$ By Lemma \ref{lem:hermitian-cross-symmetry}, we have $D_2\cap (-qL)=\varnothing.$ Also, it is easy to see that $D_1\cap(-qD_2)=\varnothing$ and equivalently, $D_2\cap(-qD_1)=\varnothing.$
It remains to consider $D_2\cap(-qD_2).$ Suppose 
$$au\equiv-qau\pmod n,$$ then $(q+1)au\equiv0\pmod{u(r+\delta-1)}.$  Since $r+\delta-1\nmid (q+1)a$, by assumption.
Therefore,
$D_2\cap(-qD_2)=\varnothing.
$
Combining the above observations, we obtain
$Z\cap(-qZ)=\varnothing.
$ Thus, $C$ is Hermitian dual-containing. As $Z$ coincides with the defining set considered in Case $1$ of Theorem \ref{thm:two-zero extension}, the resulting code $C$ is an optimal $(r,\delta)$-LRC with parameters $\bigl[u(r+\delta-1),ur-2,\delta+2\bigr]_{q^2}.
$

Moreover, as explained in Case 1 of Theorem \ref{thm:two-zero extension}, $Z^c$ contains $r$ consecutive integers. Hence, by Lemma \ref{dual distance-complement}, $d(C^{\perp_H})\geq r+1$. The hypothesis $q(\delta-1)<r$, together with $\zeta\nmid (q+1)$, implies $r>\delta+1$. Therefore, $d(C^{\perp_H})\geq r+1>\delta+2=d(C)$. Thus, the resulting quantum code $\mathcal{Q}(C)$ from Hermitian construction
is pure and has minimum distance $d_{\mathcal{Q}} =\delta+2.$ Consequently, by Theorem \ref{classical-quantum-locality}, $\mathcal{Q}(C)$ is a pure quantum
cyclic $(r, \delta)$-LRC over $\mathbb{F}_q$ with parameters
$$
\bigl[\!\bigl[
u(r+\delta-1),
u(r-\delta+1)-4,
\ \delta+2
\bigr]\!\bigr]_q.
$$
Also, it is optimal w.r.t. quantum Singleton-type bound \eqref{singleton_qlrc}.
\end{proof}
\begin{example}
Let $q=7,\ \delta=3,\ r=22,\ u=3.$ Then,  $r+\delta-1=24$ and $q^2-1=48.$ Consequently, $r+\delta-1\mid q^2-1.$ We verify the remaining hypotheses of the theorem. First,
$q(\delta-1)=7(3-1)=14<22=r.$ 
Moreover,  $\gcd(u,q^2-1)=\gcd(3,48)=3\nmid 8(=q+1),
$ and  $\gcd(u,r+\delta-1)=\gcd(3,24)=3\mid \delta.$  Choose $a=1$ and $b=0$, we have $3\cdot 1+24\cdot 0=3.$ Finally, $r+\delta-1=24\nmid (q+1)a=8\cdot 1.$ Thus, all the assumptions of the theorem are satisfied.     The locality set is
$$
L=\bigcup_{i=1}^{2}\{24m+i:0\leq m\leq2\}=\{1,2,25,26,49,50\}\pmod{72}.
$$
Since $x=\frac{u}{\zeta}=\frac{3}{3}=1$,  $D_1=\{x(r+\delta-1)\}=\{24\}.$
The second additional defining set is $D_2=\{au\}=\{3\}.$
Therefore, the complete defining set is 
$
Z=L\cup D_1\cup D_2=\{1,2,3,24,25,26,49,50\}\pmod{72}.
$
Notice that $Z$ is a union of $q^2$-cyclotomic cosets modulo $72$.
Hence,
$$-qZ=-7Z
=\{10,17,34,41,48,51,58,65\}\pmod{72}.
$$
Observe that $Z\cap(-qZ)=\varnothing.$ Using Magma \cite{magma}, we verify that $d(C)=5$ and $d(C^{\perp_H})=23$.
Hence, $C$ is an optimal cyclic $(22,3)$-LRC
with parameters $[72,64,5]_{7^2}.$
The Hermitian construction consequently produces a quantum code over
$\mathbb{F}_7$ of dimension $2k-n=2\cdot64-72=56.$ Furthermore,
$
d\left(C^{\perp_{\mathrm H}}\right)=23>d(C)=5.$ Therefore, the pure quantum cyclic $(22,3)$-LRC with parameters
$[[72,56,5]]_7$ is optimal with respect to the Singleton-type bound \eqref{singleton_qlrc}.
\end{example}

The corollary is obtained from Theorem \ref{thm:q^2-1 two zero extention} by taking $\delta=2$, except that the conditions are slightly relaxed in this special case.  

\begin{corollary}
 Let $n=u(r+1)$ such that  $(r+1)\mid(q^2-1)$,  $r>2$, $u\geq 3$, and $\zeta=\gcd(u,q^2-1).$  Assume that $(r+1)\nmid2(q+1), 
\ \zeta\nmid(q+1),$ and $\gcd(u,r+1)\mid 2.$
Then, there exists a Hermitian dual-containing optimal cyclic
$(r,2)$-locally recoverable code over $\mathbb{F}_{q^2}$ with
parameters
$[u(r+1),ur-2,4]_{q^2}$. Moreover, there exists an optimal pure quantum cyclic
$(r,2)$-LRC over $\mathbb{F}_q$ with parameters
$\bigl[\!\bigl[u(r+1),\,u(r-1)-4,\,
4\bigr]\!\bigr]_q$.
\end{corollary}
\begin{proof}
Use the construction in Theorem~\ref{thm:q^2-1 two zero extention}, with
$$
L=\{m(r+1)+1:0\leq m\leq u-1\},\quad 
D_1=\{x(r+1)\},\quad  D_2=\{au\},$$ where
$x=u/\zeta$, and $a,b$ are integers satisfying
$au+b(r+1)=2$.
The condition $(r+1)\nmid2(q+1)$ implies
$L\cap(-qL)=\varnothing$.  We next show that
$D_2\cap(-qL)=\varnothing$. Suppose on the contrary, then for some $0\leq m\leq u-1$, we have $$au\equiv-q(m(r+1)+1)\pmod{u(r+1)}.$$ Reduction modulo $r+1$ gives $q+2\equiv0\pmod{r+1}$, as $au\equiv 2\pmod{r+1}$. This implies  $q^2\equiv 4\pmod{r+1}$, hence together with $q^2\equiv1\pmod{r+1}$,  we get $(r+1)\mid3$,
contradicting $r>2$.  Hence, $D_2\cap(-qL)=\varnothing$ and by Lemma \ref{lem:hermitian-cross-symmetry}, $L\cap(-qD_2)=\varnothing$. Further, $D_2\cap (-qD_2)=\varnothing$ follows by assumption $(r+1)\nmid2(q+1)$.  All intersections involving $D_1$ are
empty by $\zeta\nmid(q+1)$ and the residue separation used in the
preceding Theorem \ref{thm:q^2-1 two zero extention}.  Hence $Z\cap(-qZ)=\varnothing$, and  further conclusion
follows from Theorem \ref{thm:q^2-1 two zero extention} by substituting $\delta=2$.
\end{proof}
\begin{theorem}\label{thm:q^2-1 general construction}
Let $n=u(r+\delta-1)$ be a positive integer such that
$(r+\delta-1)\mid(q^2-1)$ and $\gcd(u,r+\delta-1)=1.
$ Let $d$ be an integer satisfying $\delta+2\leq d\leq2\delta.$ If $r+\delta-1>(d-1)(q+1),$ then there exists a Hermitian dual-containing optimal cyclic
$(r,\delta)$-LRC over $\mathbb F_{q^2}$ with parameters
$\bigl[u(r+\delta-1),\,ur-d+\delta,\,d\bigr]_{q^2}.
$ Moreover, there exists an optimal pure quantum cyclic $(r,\delta)$-LRC
over $\mathbb F_q$ with parameters
$\bigl[\!\bigl[
u(r+\delta-1),\,
u(r-\delta+1)-2(d-\delta),\,
d
\bigr]\!\bigr]_q.
$\end{theorem}
\begin{proof}
Let
 $L=\bigcup_{i=1}^{\delta-1}
\{m(r+\delta-1)+i:0\leq m\leq u-1\}
$ and $D=\{auj: \delta\leq j\leq d-1\}$,  where
 $a,b$ are integers satisfying
$au+b(r+\delta-1)=1$. Since $(r+\delta-1)|(q^2-1)$, the set $Z=L\cup D$ is a complete defining set of a cyclic code $C$ over $\mathbb{F}_{q^2}.$ Also, $L\subset Z$, by Proposition \ref{(r,delta)locality}, $C$ is an $(r,\delta)$-LRC.  We next show that $C$ is Hermitian dual-containing. Observe that $L\cap(-qL)=\varnothing$, since $r>(d-1)(q+1)-(\delta-1)$ and $\delta+2\leq d\leq 2\delta$ implies $r>q(\delta-1)$.  Suppose that $D\cap(-qL)\neq\varnothing.$ Then there exist $m$, with $0\leq m\leq u-1$, and $i$, with $1\leq i\leq\delta-1$, for some $\delta\leq j\leq d-1$ such that   
$$auj\equiv-q(m(r+\delta-1)+i)\pmod{u(r+\delta-1)}.$$
This implies $j+qi\equiv 0 \pmod{r+\delta-1}$. However, $0<j+qi\leq (d-1)+q(\delta-1)< r+\delta-1$, where the last inequality follows from $r+\delta-1>(d-1)(q+1)$. Therefore, $D\cap(-qL)=\varnothing.$ Equivalently, $L\cap(-qD)=\varnothing.$  Moreover, $D\cap(-qD)=\varnothing$, because $r+\delta-1>(d-1)(q+1)$. Therefore, $Z\cap(-qZ)=\varnothing$.   Moreover, $Z$ coincides with the defining set considered in Theorem \ref{thm:general construction}, $C$ is a Hermitian dual-containing optimal $(r,\delta)$-LRC with parameters $[u(r+\delta-1),ur-d+\delta,d]$.

Now, we will prove the purity of the resulting quantum code. As explained in  Theorem \ref{thm:general construction},  $Z^c$ contains $r+\delta-d$ consecutive integers. By Lemma \ref{dual distance-complement}, $d(C^{\perp_{H}})\geq r+\delta-d+1$. Further the assumption $r+\delta-1>(d-1)(q+1)>2(d-1)$ implies that $d(C^{\perp_{H}})>d=d(C)$. Hence, the Hermitian construction yields a pure quantum cyclic code $\mathcal{Q}(C)$. 
Moreover, by Theorem \ref{classical-quantum-locality}, there exists pure quantum cyclic $(r,\delta)$-LRC over $\mathbb{F}_q$ with parameters
$\bigl[\!\bigl[u(r+\delta-1),\,u(r-\delta+1)-2(d-\delta),\, d\bigr]\!\bigr]_q$.
Also, it is optimal w.r.t Singleton-type bound \eqref{singleton_qlrc}
\end{proof}

\begin{remark} 
In the construction of Theorem~\ref{thm:q^2-1 general construction} the assumption
$r+\delta-1>(d-1)(q+1)
$ can be weakened when
$\zeta=\gcd(u,q^2-1)\nmid(q+1).$
Indeed, setting $Z=L\cup D_1\cup D_2$
\[
D_1=\{x(r+\delta-1)\}
,\ 
D_2=\{auj:\delta\leq j\leq d-2\},
\]
where $x=\frac{u}{\zeta}$ , yields the same classical and quantum
parameters under the weaker condition
\[
r+\delta-1>(d-2)(q+1).
\]
The proof follows from the same defining set argument. Thus, the two-set construction gives a slightly larger parameter
range, whereas Theorem~\ref{thm:q^2-1 general construction} provides
a simpler uniform formulation.
\end{remark}

\begin{example}
Let
$q=16,\ u=3,\ r=82,\ \delta=4,
\ \text{and}\ d=6.$ Then $r+\delta-1=82+4-1=85
$ and $85\mid 255=16^2-1.$ Moreover, $\gcd(u,r+\delta-1)=\gcd(3,85)=1.$ Since $
\zeta=\gcd(u,q^2-1)=\gcd(3,255)=3$ and $3\nmid17=q+1,$ 
the required divisibility condition is satisfied. Furthermore,
$\delta+2=6\leq d\leq8=2\delta,$ and $r=82>(d-2)(q+1)-(\delta-1)=65.$ Therefore, all the hypotheses of the above remark are satisfied. Since $n=u(r+\delta-1)=3\cdot85=255.$
The locality set is
$$L=\bigcup_{i=1}^{3}\{85m+i:0\leq m\leq2\}=\{1,2,3,86,87,88,171,172,173\}\pmod{255}.$$
 Since $x=\frac{u}{\zeta}=1,$ we have $D_1=\{x(r+\delta-1)\}=\{85\}.$ To determine $D_2$, choose integers $a$ and $b$ satisfying $au+b(r+\delta-1)=1.$ Since $57\cdot3-2\cdot85=1,
$ we may take $a=57,\ b=-2.$ Now,
$$
D_2=\{auj:\delta\leq j\leq d-2\}=\{57\cdot3\cdot4\}=\{174\}\pmod{255}.$$
Consequently, the complete defining set is $ Z=L\cup D_1\cup D_2=\{1,2,3,85,86,87,88,171,172,173,174\}\pmod{255}.$
In particular, $|Z|=11$; hence, $\dim(C)=255-|Z|=244.$ Further, we have $$-16Z=\{21,37,53,69,122,138,154,170,207,223,239\}\pmod{255}.
$$
Clearly, $Z\cap(-16Z)=\varnothing.$ Using Magma \cite{magma}, we have verified $d(C)=6$ and $d(C^{\perp_H})>d(C)$. Therefore, there exists a Hermitian dual-containing optimal cyclic $(82,4)$-LRC over $\mathbb{F}_{256}$ with parameters
$
[255,244,6]_{256}.
$
Moreover,  there exists a pure quantum cyclic $(82,4)$-LRC over
$\mathbb{F}_{16}$ with parameters
$\bigl[\!\bigl[255,233,d_{\mathcal Q}\geq6\bigr]\!\bigr]_{16}
$, which is optimal with respect to the bound~\eqref{singleton_qlrc}.
\end{example}

All results of this subsection are summarized in Table \ref{tab:hermitian_q2minus1}. This  table reflects the analogous progression in the Hermitian $q^2-1$
 case. 
\begin{table*}[ht]
\centering
\caption{Hermitian dual-containing optimal cyclic $(r,\delta)$-LRCs for
 $r+\delta-1\mid(q^2-1)$ and the corresponding  optimal pure quantum $(r,\delta)$-LRCs.
}
\label{tab:hermitian_q2minus1}
\renewcommand{\arraystretch}{1.35}
\setlength{\tabcolsep}{2pt}
\scriptsize
\begin{tabular}{|c|p{3.5cm}|p{3.5cm}|p{3.5cm}|p{4.6cm}|}
\hline
\textbf{Result}
&
\textbf{Conditions}
&
\textbf{ Distance Defining set}
&
\textbf{Optimal Classical LRC}
&
\textbf{Optimal Quantum LRC}
\\
\hline

Thm.\ref{thm:q^2-1 basic construction}
&
$\ \ n\mid(q^2-1)$,

$\ (q-1)(\delta-1)+\ell<r$,

$\ (q+1)\ell<n$
&
$D=\{1,2,\ldots,\ell\}$

where $\ell>\delta-1$
&
$\bigl[n,\,
ur-\ell+\delta-1,\,
\ell+1\bigr]_{q^2}$
&
$\bigl[\!\bigl[n,\,
u(r-\delta+1)-2\ell+2(\delta-1),\ell+1
\bigr]\!\bigr]_q$
\\
\hline

Thm.\ref{thm: q^2-1 one zero extention}
&
$\ (r+\delta-1)\mid(q^2-1)$,

$\ q(\delta-1)<r$,

$\ \zeta\nmid(q+1)$
&

$D=\{\dfrac{u}{\zeta}(r+\delta-1)\}$,

where $\zeta=\gcd(u,q^2-1)$
&
$\bigl[n,\,ur-1,\,\delta+1\bigr]_{q^2}$
&
$\bigl[\!\bigl[n,\,
u(r-\delta+1)-2,\ \delta+1
\bigr]\!\bigr]_q$
\\
\hline

Thm.\ref{thm:q^2-1 two zero extention}
&
$\ (r+\delta-1)\mid(q^2-1)$,

$\ q(\delta-1)<r$,

$\ \zeta\nmid(q+1)$,

$\ \gcd(u,r+\delta-1)\mid\delta$,

$\ r+\delta-1\nmid(q+1)a$&

$D_1=\{\dfrac{u}{\zeta}(r+\delta-1)\}$,

$D_2=\{au\}$,

where  $\zeta=\gcd(u,q^2-1),$

$au\equiv \delta\pmod{r+\delta-1}$
&
$\bigl[n,\,ur-2,\,\delta+2\bigr]_{q^2}$
&
$\bigl[\!\bigl[n,\,
u(r-\delta+1)-4, \delta+2
\bigr]\!\bigr]_q$
\\
\hline

Thm.\ref{thm:q^2-1 general construction}
&
$\ (r+\delta-1)\mid(q^2-1)$,

$\ \gcd(u,r+\delta-1)=1$,

$\ \delta+2\leq d\leq2\delta$,

$\ (d-1)(q+1)-(\delta-1)<r$
&

$D=\{auj:\delta\leq j\leq d-1\}$,

$au\equiv 1\pmod{r+\delta-1}$
&
$\bigl[n,\,ur-d+\delta,\,d\bigr]_{q^2}$
&
$\bigl[\!\bigl[n,\,
u(r-\delta+1)-2(d-\delta), d
\bigr]\!\bigr]_q$
\\
\hline
\end{tabular}
\end{table*}
\subsection{Hermitian dual-containing optimal cyclic $(r,\delta)$-LRCs when $(r+\delta-1)|(q^2+1)$}\label{subsec:q2-plus-one}
Throughout this subsection, $n=u(r+\delta-1)$ and $\delta\geq3$ are odd integers, where $u\geq 3,\ r\geq 3.$   Define
\begin{equation}
L
=
\bigcup_{\substack{1\leq j\leq\delta-2\\j\ {\rm odd}}}
\{m(r+\delta-1)\pm j:0\leq m\leq u-1\}.
\label{eq:symmetric-hermitian-L}
\end{equation}
The residues modulo $(r+\delta-1)$ are $
\pm1,\pm3,\ldots,\pm\delta-4,\pm\delta-2
$. These are in the form of an arithmetic progression of length $\delta-1$ and common difference
$2$.  Since $n$ is odd, Proposition~\ref{(r,delta)locality} shows
that $L$ is a valid locality set  with locality $(r,\delta).$

In this subsection, we construct four families of Hermitian dual-containing cyclic
$(r,\delta)$-LRCs in the case
$r+\delta-1\mid(q^2+1),
$ and consequently obtain corresponding families of  optimal pure quantum $(r,\delta)$-LRCs over $\mathbb F_q$.
This case is structurally different from the Euclidean setting. In particular, when $r+\delta-1\mid(q+1),$  the locality defining set necessarily exhibits a symmetric residue structure, so that opposite residues occur simultaneously. Consequently, $L\cap(-L)\neq\varnothing,$
which prevents the corresponding cyclic code from being Euclidean dual-containing. In contrast, Hermitian dual containment is governed by the condition
$Z\cap(-qZ)=\varnothing,$ and this additional multiplication by $q$ makes it possible to construct Hermitian
dual-containing cyclic $(r,\delta)$-LRCs in this parameter regime.
The following lemma will be used throughout the constructions below.
\begin{lemma}\label{q2+1 Lcyclotomic coset}
 Let $n=u(r+\delta-1)$ and $\delta$ be  odd integers and assume that $(r+\delta-1)\mid(q^2+1)$. Define
\[
L=\bigcup_{j=1}^{\frac{\delta-1}{2}}\{m(r+\delta-1)\pm\ell_j:0\le m\le u-1\},\]
where $0<\ell_j<(r+\delta-1).$ Then, $L$ is a union of $q^2$-cyclotomic cosets modulo $n$. 
\end{lemma}
\begin{proof}
Since $(r+\delta-1)\mid(q^2+1)$, we have $q^2\equiv -1 \pmod {r+\delta-1}.$ Then
\begin{align*}
 q^2(m(r+\delta-1)+\ell_j)&=(-1+a(r+\delta-1))(m(r+\delta-1)+\ell_j)\\
 &\equiv -\ell_j\pmod{(r+\delta-1)}\\
 &\equiv m'(r+\delta-1)-\ell_j\pmod{u(r+\delta-1)}
 \end{align*}
 for some $0\leq m'\leq u-1$. Hence, $L$ is a union of $q^2$-cyclotomic cosets $n$. 
\end{proof}
\begin{theorem}\label{thm:q2+1 basic construction}
Let $n=u(r+\delta-1)$ and $\delta\geq 3$ be  odd integers such that $n\mid(q^2+1)$.  Suppose that for some odd integer $\ell>\delta-2$, we have $q(\delta-2)+\ell<r+\delta-1 \text{ and } (q+1)\ell<n.$ Then, there exists a Hermitian dual-containing optimal cyclic $(r,\delta)$-LRC with parameters $\bigl[u(r+\delta-1),\, ur-\ell+(\delta-2),\, \ell+2\bigr]_{q^2}.$ Moreover, there exists an optimal pure quantum cyclic $(r,\delta)$-LRC with parameters $\bigl[\!\bigl[
    u(r+\delta-1),\,
    u(r-\delta+1)-2\ell+2(\delta-2),\,
    \ell+2
    \bigr]\!\bigr]_q.$  
\end{theorem}

\begin{proof}
Let $D=\{\pm1,\pm3,\ldots,\pm\ell\},$
where $\ell>\delta-2$.
Since $n\mid(q^2+1)$, 
$L$ and $D$ are union of $q^2$-cyclotomic coset modulo $n$. In particular, $Z=L\cup D$ defines a cyclic code $C$ over $\mathbb{F}_{q^2}$ and by Proposition \ref{(r,delta)locality}, $C $ has $(r,\delta)$ locality. Now, we prove that $Z\cap(-qZ)=\varnothing$ by considering the four possible intersections
\[
L\cap(-qL),\qquad
L\cap(-qD),\qquad
D\cap(-qL),\qquad
D\cap(-qD).
\]
First, every element of $L$ is congruent modulo $r+\delta-1$ to one of $\pm1,\pm3,\ldots,\pm(\delta-2).$ Suppose that $L\cap(-qL)\neq\varnothing$. Then there exist odd integers $a,b$ satisfying $1\leq |a|,|b|\leq\delta-2$ such that
\[
a\equiv-qb\pmod{r+\delta-1}.
\]
On the other hand, $|qb+a|
\leq q|b|+|a|
\leq q(\delta-2)+(\delta-2)
=(q+1)(\delta-2).$ Since $q(\delta-2)+\ell<r+\delta-1$ and $\ell>\delta-2$,
we have
\[
(q+1)(\delta-2)<q(\delta-2)+\ell<r+\delta-1;
\]
 hence, $|qb+a|<r+\delta-1$, implying $qb+a=0$. However, $\ell<q$, which follows from
$(q+1)\ell<n$ and $n|(q^2+1)$, we have  $q>\ell>\delta-2\geq |a|$ and $b\neq0$, so
$|qb|>|a|$, hence $qb+a\neq 0$. Thus
$L\cap(-qL)=\varnothing.$  Next, suppose that $D\cap(-qL)\neq\varnothing$. Then there exist an element
$d\in D$ and an odd integer $b$ with $1\leq |b|\leq\delta-2$ such that $d\equiv-qb\pmod{r+\delta-1}.$
Since $|d|\leq\ell$, we have
\[
|qb+d|
\leq q(\delta-2)+\ell<r+\delta-1,
\] and consequently $qb+d=0$, which is impossible since 
$|qb|\geq q>\ell\geq|d|$. Hence
$D\cap(-qL)=\emptyset.$ Similarly, $L\cap(-qD)=\varnothing$. Finally, suppose that $D\cap(-qD)\neq\varnothing.$ Then, there exist $x,y\in D$ such that
\[
x\equiv-qy\pmod n.
\]
Since $|x|,|y|\leq\ell$, we have $|qy+x|\leq q\ell+\ell=(q+1)\ell<n.$ Further $q>\ell$ and $y\neq0$, we have $|qy|\geq q>\ell\geq|x|,$ hence $qy+x\neq0$. Therefore $D\cap(-qD)=\varnothing.$  Consequently, $Z\cap(-qZ)=\varnothing.$ Hence, $C$ is a Hermitian dual-containing code. We now determine the parameters of $C$. Since
$D$ already contains the elements $\pm1,\ldots,\pm\delta-2$ corresponding to
the $m=0$ part of $L$, we have
$|Z|=(u-1)(\delta-1)+(\ell+1)$. Therefore,
\[
\dim(C)
=n-|Z|
=u(r+\delta-1)-(u-1)(\delta-1)-(\ell+1)
=ur-\ell+(\delta-2).
\]
Furthermore, $Z$ contains $\pm1,\pm3,\ldots,\pm\ell$ consecutive integers with common difference $2$. Since $n$ is odd, the BCH bound gives
$d(C)\geq\ell+2$. On the other hand, applying
bound \eqref{singleton_lrc} yields $d(C)\leq\ell+2$. Thus,
$d(C)=\ell+2$. Consequently, $C$ is a Hermitian dual-containing optimal cyclic
$(r,\delta)$-LRC with parameters
\[
\bigl[u(r+\delta-1),\,ur-\ell+(\delta-2),\,\ell+2\bigr]_{q^2}.
\]
We finally establish the purity of the resulting quantum code. Consider the gap between the two consecutive locality blocks corresponding to \(m=1\) and \(m=2\), and define
$$
A=\bigl\{
(r+\delta-1)+(\delta-1),\,
(r+\delta-1)+\delta,\,
\ldots,\,
2(r+\delta-1)-(\delta-1)
\bigr\}.
$$
Clearly, none of the elements of \(A\) belongs to the locality set \(L\). Moreover, the additional defining set \(D\) does not intersect \(A\). Hence, \(A\subseteq Z^c\). Since \(A\) consists of \(r-\delta+2\) consecutive integers, Lemma~\ref{dual distance-complement} gives $d\bigl(C^{\perp_H}\bigr)\geq r-\delta+3.
$ Furthermore, since \((q+1)\ell<n\leq q^2+1\), we obtain \(\ell<q\). Since both \(\ell\) and \(\delta\) are odd and \(\ell>\delta-2\), it follows that \(\ell\geq\delta\). Thus, \(q>\ell\geq\delta\), and consequently \(q\geq\delta+1\). Therefore, for $\delta\geq 3$
$$
q(\delta-2)\geq(\delta+1)(\delta-2)\geq 2\delta-2.
$$
Combining this with \(q(\delta-2)+\ell<r+\delta-1\), we obtain $r>\ell+\delta-1$. Hence $$d\bigl(C^{\perp_H}\bigr)
\geq r-\delta+3
>\ell+2=d(C).
$$
Thus, the quantum code \(\mathcal{Q}(C)\) obtained from the Hermitian construction is pure and has minimum distance
\(d_{\mathcal Q}=\ell+2\). Finally, by Theorem~\ref{classical-quantum-locality}, \(\mathcal{Q}(C)\) is a quantum \((r,\delta)\)-LRC over \(\mathbb{F}_q\) with parameters
$
\left[\!\left[
u(r+\delta-1),\,
u(r-\delta+1)-2\ell+2(\delta-2),\,
\ell+2
\right]\!\right]_q.
$ 
Moreover, it attains the quantum Singleton-type bound~\eqref{singleton_qlrc}, and hence is optimal.
\end{proof}
\begin{example} 

Let $q=17$, $u=5$, $r=27$, $\delta=3$, and $\ell=7$. Then
$r+\delta-1=29$ and $n=u(r+\delta-1)=145$. Both $n$ and $\delta$ are odd, and $145\mid(17^2+1)$, i.e., $n|(q^2+1)$ Moreover, $\ell=7>\delta-2=1$, and $q(\delta-2)+7=17+7=24<29=r+\delta-1$, $(q+1)\ell=18\cdot7=126<145=n.$
Thus, all the conditions of Theorem \ref{thm:q2+1 basic construction} are satisfied. The locality set is $$
L=\{29m\pm1:0\leq m\leq 4\}
=\{1,28,30,57,59,86,88,115,117,144\}\pmod{145}.
$$
The distance set is $D=\{\pm1,\pm3,\pm5,\pm7\}
=\{1,3,5,7,138,140,142,144\} \pmod{145}.$
Consequently, the complete defining set is
$
Z=L\cup D
=\{1,3,5,7,28,30,57,59,86,88,115,117,138,140,142,144\}\pmod{145}.
$
Furthermore, $-17Z=\{12,17,26,41,46,51,60,70,75,85,94,99,104,119,128,133\}\pmod{145}.
$
Hence, $Z\cap(-17Z)=\varnothing.$ Using Magma \cite{magma}, we find that $d(C)=9$ and $d(C^{\perp_H})=28$. Therefore, there exists a Hermitian dual-containing optimal cyclic $(27,3)$-LRC over $\mathbb{F}_{289}$ with parameters $[145,129,9]_{289}.$ Since  $d(C^{\perp_H})>d(C)$, the corresponding pure quantum cyclic $(27,3)$-LRC with parameters
$[[145,113,9]]_{17}$ is optimal.
\end{example}
\begin{theorem}\label{thm:q2+1 two zero extention}
Let $n=u(r+\delta-1)$ be an odd integer and $(r+\delta-1)\mid (q^2+1)$, where $\delta\geq3$ is an odd integer. Suppose that
$$
\gcd(u,r+\delta-1)\mid\delta,\ \ 
q>\delta, \text{ and }\  r>q(\delta-2)+1.
$$
Then there exists a Hermitian dual-containing optimal cyclic $(r,\delta)$-LRC over $\mathbb F_{q^2}$ with parameters
$
\bigl[u(r+\delta-1),ur-2,\delta+2\bigr]_{q^2}.
$
Moreover, there exists a  pure optimal quantum cyclic $(r,\delta)$-LRC over $\mathbb F_q$ with parameters
$
\bigl[\bigl[
u(r+\delta-1),
u(r-\delta+1)-4,\delta+2
\bigr]\bigr]_q.
$
\end{theorem}

\begin{proof}
Since $\gcd(u,(r+\delta-1))\mid\delta$, there exist integers $a,b$ satisfying
$
au+b(r+\delta-1)=\delta.
$
Consider the cyclic code $C$ over $\mathbb F_{q^2}$ with defining set $Z=L\cup D,$ where
\begin{align*}
L=
\bigcup_{\substack{j=1\\ j\ {\rm odd}}}^{\delta-2}
\{m(r+\delta-1)\pm j:0\leq m\leq u-1\},\ \ \ \
D=\{au,-au\}.
\end{align*}
We first show that $Z$ is a union of $q^2$-cyclotomic cosets modulo $n$. By Lemma~\ref{q2+1 Lcyclotomic coset}, $L$ is a union of $q^2$-cyclotomic cosets modulo $n$. Moreover, since $(r+\delta-1)\mid q^2+1$, we have
$$
(q^2+1)au
\equiv0\pmod{u(r+\delta-1)}.
$$
Hence $q^2au\equiv-au\pmod n$ and, similarly, $q^2(-au)\equiv au\pmod n.
$ Thus $\{au,-au\}$ is a $q^2$-cyclotomic coset modulo $n$. Consequently, $Z$ is a union of $q^2$-cyclotomic cosets; hence, a complete defining set of a cyclic code $C$ over $\mathbb F_{q^2}$. Since $L\subseteq Z$, Proposition~\ref{(r,delta)locality} implies that $C$ is a cyclic $(r,\delta)$-LRC. We next show that $C$ is Hermitian dual-containing.  Since $q(\delta-2)+1<r$, we can show that $L\cap(-qL)=\varnothing$, as shown in Theorem \ref{thm:q2+1 basic construction}. We now consider the intersection between $D$ and $-qL$.  Suppose that
$D\cap(-qL)\neq\varnothing.$ Then, for some odd integer $j$ with $1\leq|j|\leq\delta-2$ and some $\varepsilon\in\{1,-1\}$,
$$
\varepsilon\delta\equiv-qj\pmod{r+\delta-1}.
$$
However, $
|qj+\varepsilon\delta|\leq q(\delta-2)+\delta<r+\delta-1,$ where the last inequality follows from
$r>q(\delta-2)+1$. Therefore, $qj+\varepsilon\delta=0.$ This is impossible because $|qj|\geq q>\delta.$ Consequently, $D\cap(-qL)=\varnothing.$ Similarly, $L\cap(-qD)=\varnothing.$ It remains to show that $D\cap(-qD)=\varnothing.$ Suppose on contrary, $au\equiv \pm qau\pmod{n}$. Reducing modulo $r+\delta-1$ yields
$$
\delta\equiv\pm q\delta\pmod {r+\delta-1},
$$
and hence $(r+\delta-1)\mid\delta(q\pm1).$ Since $r+\delta-1$ is odd and $r+\delta-1\mid (q^2+1)$, we have
$$
\gcd(r+\delta-1,q-1)=\gcd(r+\delta-1,q+1)=1.
$$
Indeed, any common divisor of $r+\delta-1$ and $q\pm1$ divides both $q^2+1$ and $q^2-1$, and therefore divides $2$; since $r+\delta-1$ is odd, such a divisor must be $1$. It follows that $r+\delta-1\mid\delta$, which is impossible because $0<\delta<r+\delta-1.$
Therefore,
$
D\cap(-qD)=\varnothing.
$
Combining the above intersections gives $Z\cap(-qZ)=\varnothing.$ Hence $C$ is Hermitian dual-containing. We now determine the parameters of $C$. The sets $L$ and $D$ are disjoint, thus $|Z|=u(\delta-1)+2,
$ and hence
$$
\dim(C)
=n-|Z|=u(r+\delta-1)-u(\delta-1)-2\
=ur-2.$$ Furthermore, modulo $r+\delta-1$, the defining set $Z$ contains the arithmetic progression
$$
-\delta,-(\delta-2),\ldots,-3,-1,1,3,\ldots,\delta-2,\delta,
$$
which consists of $\delta+1$ terms with common difference $2$. Since $(r+\delta-1)$ is odd, $\gcd(2,r+\delta-1)=1.$ Therefore, by Lemma~\ref{distance}, we obtain $d(C)\geq\delta+2.$ On the other hand, by the Singleton-type bound for $(r,\delta)$-LRCs, $d(C)\leq \delta+2.$
Hence $d(C)=\delta+2.$ Therefore, $C$ is an optimal Hermitian dual-containing cyclic $(r,\delta)$-LRC with parameters $
\bigl[u(r+\delta-1),ur-2,\delta+2\bigr]_{q^2}.
$
We finally establish the purity of the resulting quantum code. Since between two consecutive locality blocks, the set
$$
A_m=\{m(r+\delta-1)+(\delta-1),\,m(r+\delta-1)+\delta,\ldots,(m+1)(r+\delta-1)-(\delta-1)\}
$$
contains \(r-\delta+2\) consecutive integers which do not belong to
\(L\). Since \(D=\{au,-au\}\) contains only two elements and \(u\geq3\),
there exists at least one gap \(A_m\) such that \(A_m\cap D=\varnothing\).
Consequently, $A_m\subseteq Z^c.$ Therefore, by Lemma~\ref{dual distance-complement}, $
d\bigl(C^{\perp_H}\bigr)\geq r-\delta+3.
$ It remains to compare this bound with \(d(C)=\delta+2\). Since
\(q>\delta\), we have \(q\geq\delta+1\), and hence by assumption
$$
r>q(\delta-2)+1
\geq(\delta+1)(\delta-2)+1
\geq2\delta-1
$$
for every odd \(\delta\geq3\). Consequently,
$
d\bigl(C^{\perp_H}\bigr)
\geq r-\delta+3
>\delta+2=d(C).
$
Hence, the quantum code \(\mathcal{Q}(C)\) obtained from the Hermitian
construction is pure and has minimum distance
$
d_{\mathcal Q}=\delta+2.
$
Therefore, by Theorem~\ref{classical-quantum-locality},
\(\mathcal{Q}(C)\) is a pure quantum cyclic \((r,\delta)\)-LRC over
\(\mathbb F_q\) with parameters
$$
\left[\!\left[
u(r+\delta-1),\,
u(r-\delta+1)-4,\,
\delta+2
\right]\!\right]_q.
$$
Since it attains the quantum Singleton-type bound
\eqref{singleton_qlrc}, the resulting quantum code is optimal.
\end{proof}
\begin{example}
Let $q=11,\ \delta=5,\ r=57,\ u=3.$ Then, $r+\delta-1=57+5-1=61$ and $q^2+1=11^2+1=122.$ Consequently, $r+\delta-1=61\mid122=q^2+1.$ Moreover, $n=u(r+\delta-1)=3\cdot61=183$ is odd, and $\gcd(u,r+\delta-1)=\gcd(3,61)=1,$ i.e., $\gcd(u,r+\delta-1)\mid5=\delta.$ Furthermore,
$q=11>5=\delta$ and $r=57>11(5-2)+1=34.$ Thus, all the hypotheses of the theorem are satisfied.
The locality set is
$$
L=\bigcup_{m=0}^{2}\{61m\pm1,61m\pm3\}=\{1,3,58,60,62,64,119,121,123,125,180,182\}\pmod{183}.
$$
Choose $a=22,\ b=-1,$ so that $au+b(r+\delta-1)=22\cdot3-61=5=\delta.$ The additional defining set is
$
D=\{\pm au\}=\{\pm66\}=\{66,117\}\pmod{183}.
$
Therefore, the complete defining set is
$$
Z=L\cup D=\{
1,3,58,60,62,64,66,117,119,121,123,125,180,182
\}\pmod{183}.
$$
Moreover, $-qZ=-11Z=\{6,11,28,33,50,72,89,94,111,133,150,155,172,177\}.$ Hence, $Z\cap(-qZ)=\varnothing.$
Since $|Z|=14$, the dimension of $C$ is $k=n-|Z|=183-14=169=ur-2$. From Magma \cite{magma} computations, $d(C)=7$ and $d(C^{\perp_H})=58.$ Hence, $C$ is a Hermitian dual-containing optimal cyclic $(57,5)$-LRC  with
parameters $[183,169,7]_{121}.$
Moreover, $d\left(C^{\perp_{\mathrm H}}\right)>d(C),$ the resulting quantum code by Hermitian construction is pure and has minimum distance $7$. Hence, there exists a pure  optimal quantum cyclic $(57,5)$-LRC with parameters $[[183,155,7]]_{11}.$
\end{example}
\begin{theorem}\label{thm:q2+1 2delta construction}
Let $n=u(r+\delta-1)$ be an odd integer such that $(r+\delta-1)\mid(q^2+1)$. Suppose that
$\gcd(u,r+\delta-1)=1$, $\gcd(u,q^2-1)\nmid q+1$. If  $r>q(\delta -1)$, and $q>\delta$, then there exists a Hermitian dual-containing optimal cyclic
$(r,\delta)$-LRC with parameters
$\bigl[u(r+\delta-1),\,ur-\delta,\,2\delta\bigr]_{q^2}.$
Moreover, there exists an optimal pure quantum cyclic $(r,\delta)$-LRC with parameters
$\bigl[\!\bigl[u(r+\delta-1),\,
u(r-\delta+1)-2\delta,\, 2\delta\bigr]\!\bigr]_q.$
\end{theorem}
\begin{proof}
Let
$$
 D_1=\{x(r+\delta-1)\},\quad
 D_2=\{\pm 2auj: 1\leq j\leq \frac{\delta-1}{2}\},   
$$ where
$x=\frac{u}{\zeta}$, $\zeta=\gcd(u,q^2-1)$, and $a,b$ are integers satisfying
$au+b(r+\delta-1)=1$. First, we show that $Z=L\cup D_1\cup D_2$ is a union of $q^2$-cyclotomic cosets modulo $n$.
By Lemma \ref{q2+1 Lcyclotomic coset}, $L$ is union of $q^2$-cyclotomic cosets modulo $n$. Next, since $\zeta\mid q^2-1$,
\[
(q^2-1)\frac{u}{\zeta}(r+\delta-1)\equiv0\pmod{u(r+\delta-1)}.
\]
Thus $q^2x(r+\delta-1)\equiv x(r+\delta-1)\pmod n,$ and consequently $D_1=\{x(r+\delta-1)\}$ is a singleton $q^2$-cyclotomic coset. Since $(r+\delta-1)|(q^2+1)$, for $1\leq j\leq(\delta-1)/2$, we obtain
$$
(q^2+1)(2auj)\equiv 0\pmod{u(r+\delta-1)}$$
and hence $q^2(2auj)\equiv-2auj\pmod n.$ Similarly, $q^2(-2auj)\equiv2auj\pmod n.$ Therefore $\{2auj,-2auj\}$ is a $q^2$-cyclotomic coset modulo $n$ for $1\leq j\leq(\delta-1)/2$. Hence, $D_2=\{\pm2au,\pm4au,\ldots,\pm(\delta-1)au\}
$ is a union of $q^2$-cyclotomic cosets. Thus, $Z$ is a  complete defining
set of a cyclic code $C$ over $\mathbb F_{q^2}$ and $L\subset Z$ by Proposition \ref{(r,delta)locality}, $C$ is a cyclic $(r,\delta)$ LRC over $\mathbb{F}_{q^2}.$

We now prove that the corresponding cyclic code $C$ is Hermitian dual-containing.  As shown in Theorem \ref{thm:q2+1 basic construction},
$L\cap(-qL)=\varnothing$ under the assumptions
$q>\delta$ and $(q+1)(\delta-2)<r+\delta-1$. Modulo $r+\delta-1$, the elements of $L$, $D_1$, and $D_2$ have residues
\[
\{\pm1,\pm3,\ldots,\pm(\delta-2)\},\qquad
\{0\},\qquad
\{\pm2,\pm4,\ldots,\pm(\delta-1)\},
\]
respectively. Since multiplication by $-q$ preserves the zero residue modulo $r+\delta-1$, while $L$ and $D_2$ contain only nonzero residues, we immediately obtain
$L\cap(-qD_1)=\varnothing$ and
$D_1\cap(-qD_2)=\varnothing$.

Next, we consider the intersection $D_1\cap(-qD_1)$. Observe that  under the assumption $\gcd(u,q^2-1)\nmid q+1,$ we have
$D_1\cap(-qD_1)=\varnothing.$ 
It remains to consider the intersections involving $L$ and $D_2$. Suppose, to the contrary, that
$D_2\cap(-qL)\neq\varnothing$. Then there exist an even integer $i$ and an odd integer $j$ such that
$2\leq |i|\leq\delta-1$ and $1\leq |j|\leq\delta-2$, satisfying
\[
i\equiv-qj\pmod{r+\delta-1}.
\]
On the other hand, using
$q(\delta-1)<r$, we obtain
$|qj+i|
\leq q(\delta-2)+(\delta-1)
<r+\delta-1.
$ Hence, necessarily $qj+i=0$. This is impossible, since $|qj|\geq q>\delta> |i|.$ Therefore, $D_2\cap(-qL)=\varnothing$. Finally, suppose that
$D_2\cap(-qD_2)\neq\varnothing$. Then, there exist
$i,j\in\{\pm2,\pm4,\ldots,\pm(\delta-1)\}$ such that
\[
j\equiv-qi\pmod{r+\delta-1}.
\]
However, $|qi+j|
\leq q(\delta-1)+(\delta-1) <r+\delta-1.$
 Hence, necessarily $qi+j=0$, which is impossible because
$|qi|>q>\delta>|j|.$ 
Thus,
$D_2\cap(-qD_2)=\varnothing$. Combining the above observations, we conclude that $Z\cap(-qZ)=\varnothing.$ 
Hence, $C$ is Hermitian dual-containing. 
We next determine the parameters of $C$. The sets $L$, $D_1$, and $D_2$ are pairwise disjoint. Therefore, $|Z|=u(\delta-1)+1+(\delta-1)=u(\delta-1)+\delta.$ It follows that
\[
\dim(C)
=n-|Z|
=u(r+\delta-1)-u(\delta-1)-\delta
=ur-\delta.
\]
Furthermore modulo $r+\delta-1$, $Z$ contains the set
$$
\{0,\pm1,\pm2,\pm3,\ldots,\pm(\delta-1)\}
\pmod{r+\delta-1}.
$$
Hence, by Lemma \ref{distance}, $d(C)\geq2\delta.$ On the other hand, the Singleton-type bound  \eqref{singleton_lrc} gives $d(C)\leq2\delta.$ Thus, $C$ is an optimal Hermitian dual-containing cyclic $(r,\delta)$-LRC with parameters
\[
[u(r+\delta-1),\,ur-\delta,\,2\delta]_{q^2}.
\]
We finally establish the purity of the resulting quantum code. From the above description of the residues of $Z$,  for any \(0\leq m\leq u-1\), the set
$$
A_m=
\{m(r+\delta-1)+\delta,\,m(r+\delta-1)+\delta+1,\ldots,(m+1)(r+\delta-1)-\delta\}
$$
is contained in \(Z^c\). The set \(A_m\) consists of
$r-\delta$ consecutive integers. Therefore, by Lemma~\ref{dual distance-complement}, $d\bigl(C^{\perp_H}\bigr)\geq r-\delta+1.$
It remains to compare this lower bound with \(d(C)=2\delta\). Since \(q>\delta\), we have \(q\geq\delta+1\). Thus, for $\delta\geq 3$
$$
r>q(\delta-1)\geq(\delta+1)(\delta-1)\geq3\delta-1.
$$
Consequently, $
d\bigl(C^{\perp_H}\bigr)
\geq r-\delta+1
>2\delta=d(C).
$
Therefore, the quantum code \(\mathcal{Q}(C)\) obtained from the Hermitian construction is pure and has minimum distance
$
d_{\mathcal Q}=2\delta.
$
By Theorem~\ref{classical-quantum-locality}, \(\mathcal{Q}(C)\) is a pure quantum cyclic \((r,\delta)\)-LRC over \(\mathbb F_q\) with parameters
$$
\left[\!\left[
u(r+\delta-1),\,
u(r-\delta+1)-2\delta,\,
2\delta
\right]\!\right]_q.
$$
Since it attains the quantum Singleton-type bound~\eqref{singleton_qlrc}, the resulting quantum code is optimal.
\end{proof}
\begin{example}
Let $q=4$, $u=3$, $r=15$, and $\delta=3$. Then
$r+\delta-1=17$ and $n=51$. Notice that $n$ and $\delta$ are
odd, and $r+\delta-1=17\mid(4^2+1),\ 
\gcd(u,r+\delta-1)=\gcd(3,17)=1,\ \gcd(u,q^2-1)=\gcd(3,15)=3$ i.e., $3\nmid 5=q+1.$ Furthermore, $q=4>\delta=3$, and $r=15>4(3-1)=8.$ The locality set is $$L=\{17m\pm 1: 0\leq m\leq 2\}=\{1,16,18,33,35,50\}\subseteq\mathbb{Z}_{51}$$
Let $x=\frac{u}{\zeta}=1$, set $D_1=\{x(r+\delta-1)\}=\{17\}$. Choose $a=6$ and $b=-1$, so $au+b(r+\delta-1)=6\cdot3+(-1)\cdot17=1$. Then $$D_2=\{\pm 2auj:1\leq j\leq 1\}=\{\pm 36\}=\{15,36\}\subseteq\mathbb{Z}_{51}.$$
Therefore, the complete defining set is $Z=L\cup D_1\cup D_2=\{1,15,16,17,18,33,35,36,50\}\subseteq\mathbb{Z}_{51}$. Since $q=4$, we have $-4Z=\{4,9,13,21,30,34,38,42,47\}\pmod{51}.$ Hence, $Z\cap(-4Z)=\varnothing.$  Using Magma \cite{magma}, we verify that $d(C)=6$ and $d(C^{\perp_H})=16$. Consequently, there exists a Hermitian
dual-containing optimal cyclic $(15,3)$-LRC over $\mathbb{F}_{16}$
with parameters
$
[51,42,6]_{16}.
$
Since $d(C^{\perp_H})=16>d(C)$, the corresponding pure quantum cyclic $(15,3)$-LRC with parameters
$
[[51,33,6]]_{4}
$ is optimal.
\end{example}
The following corollary provides a family of optimal $(r,\delta)$-LRCs with minimum distance $2\delta-2$ under weaker assumptions than those required in Theorem~\ref{thm:q2+1 2delta construction}. \begin{corollary}\label{cor:q2+1 2delta-2 construction}
Let $n=u(r+\delta-1)$ be an odd integer such that $(r+\delta-1)\mid(q^2+1)$. Suppose that
$\gcd(u,r+\delta-1)=1$, $\gcd(u,q^2-1)\nmid q+1$. If $q>(\delta-2)
\text{ and } r>q(\delta-2)-1$ then there exists a Hermitian dual-containing optimal cyclic $(r,\delta)$-LRC over $\mathbb F_{q^2}$ with parameters
$\bigl[u(r+\delta-1),ur-\delta+2,2\delta-2\bigr]_{q^2}.
$
Moreover, there exists an optimal pure quantum cyclic $(r,\delta)$-LRC over $\mathbb F_q$ with parameters
$
\bigl[\bigl[u(r+\delta-1),
u(r-\delta+1)-2(\delta-2),
2\delta-2
\bigr]\bigr]_q.
$
\end{corollary}

\begin{proof}
The result follows by modifying the construction in the  Theorem \ref{thm:q2+1 2delta construction} and taking
$$
D_2=\{\pm2auj:1\leq j\leq\frac{\delta-3}{2}\}.
$$
The same argument shows that $Z=L\cup D_1\cup D_2$ is a union of $q^2$-cyclotomic cosets modulo $n$ and that $Z\cap(-qZ)=\varnothing$ under the assumptions
$q>\delta-2$ and $r>q(\delta-2)-1$. Hence, the resulting cyclic code is Hermitian dual-containing. 
Moreover, $Z$ contains elements whose residues modulo $r+\delta-1$ are
$
\{0,\pm1,\pm2,\ldots,\pm(\delta-2)\}.
$
Thus, $Z$ contains $(2\delta-3)$ consecutive integers modulo $(r+\delta-1)$. Therefore, by Lemma \ref{distance}, we obtain
$
d(C)\geq 2\delta-2.
$  The Singleton-type bound for $(r,\delta)$-LRCs yields
$d(C)\leq2\delta-2$; hence, $d(C)=2\delta-2.
$ Therefore, $C$ is a Hermitian dual-containing optimal cyclic $(r,\delta)$-LRC with parameters $\bigl[u(r+\delta-1),ur-\delta+2,2\delta-2\bigr]_{q^2}.$ 
Observe that, none of the residues
$
\delta-1,\delta,\ldots,(r+\delta-1)-(\delta-1)
$
occur in \(Z\). Therefore, for every \(0\leq m\leq u-1\),
$$
A_m=
\{m(r+\delta-1)+\delta-1,\,m(r+\delta-1)+\delta,\ldots,(m+1)(r+\delta-1)-(\delta-1)\}
\subseteq Z^c.
$$
The set \(A_m\) consists of $r-\delta+2$
consecutive integers. Hence, by Lemma~\ref{dual distance-complement}, $d(C^{\perp_H})\geq r-\delta+3.
$  Since \(q>\delta-2\), we have \(q\geq\delta-1\). For \(\delta\geq5\), 
$$
r>q(\delta-2)-1
\geq(\delta-1)(\delta-2)-1\geq3(\delta-1)-1>3(\delta-1)-2=3\delta-5.
$$
For \(\delta=3\), the hypothesis \(q>\delta-2\) gives \(q>1\). The cases \(q=2\) and \(q=3\) are incompatible with \(\gcd(u,q^2-1)\nmid(q+1)\) together with the fact that \(u\) is odd. Hence \(q\geq4\). Thus, for $\delta=3$
$$
r>q(\delta-2)-1\geq4(3-2)-1=3.
$$
Since \(r+\delta-1=r+2\) is odd, $r$ is odd and thus \(r\geq5\). Consequently, in all admissible cases,
$
r>3\delta-5.
$
It follows that $
d(C^{\perp_H})
\geq r-\delta+3
>2\delta-2=d(C).
$
Consequently, by Theorem~\ref{classical-quantum-locality},
\(\mathcal Q(C)\) is an optimal pure quantum cyclic \((r,\delta)\)-LRC over
\(\mathbb F_q\) with parameters
$
\left[\!\left[
u(r+\delta-1),\,
u(r-\delta+1)-2(\delta-2),\,
2\delta-2
\right]\!\right]_q.
$
\end{proof}
\begin{example}
Let $q=7$, $u=3$, $r=21$, and $\delta=5$. Then
$r+\delta-1=25$ and $n=75$. Both $n$ and $\delta$ are odd, and
$25\mid(7^2+1)$. Moreover,
\[
\gcd(u,r+\delta-1)=\gcd(3,25)=1,
\qquad
\zeta=\gcd(u,q^2-1)=\gcd(3,48)=3\nmid8(=q+1).
\]
The remaining inequalities are $r=21>q(\delta-2)-1=7(5-2)-1=20,\ q=7>\delta-2=3.$
For $\delta=5$, the locality set is
$$L=\{25m\pm1:0\leq m\leq2\}\cup\{25m\pm3:0\leq m\leq2\} =\{1,3,22,24,26,28,47,49,51,53,72,74\}\subseteq\mathbb{Z}_{75}.$$
Let $x=\frac{u}{\zeta}=1.$ Therefore, $D_1=\{x(r+\delta-1)\}=\{25\}.$ Choose $a=-8$ and $b=1$, so that $au+b(r+\delta-1)=(-8)\cdot3+1\cdot25=1.$ Since $(\delta-3)/2=1$, the second additional defining set is
\[
D_2
=
\{\pm2auj:1\leq j\leq1\}
=
\{\pm2(-8)(3)\}
=
\{\pm48\}
=
\{27,48\}
\pmod {75}.
\]
 Consequently, the complete defining set is
$
Z=\{1,3,22,24,25,26,27,28,47,48,49,51,53,72,74\}\pmod{75}.
$ 
Since $q=7$, we compute $-qZ=-7Z$ modulo $75$. More precisely,
\[
-7Z=\{4,7,18,21,29,32,36,39,43,46,50,54,57,68,71\}\pmod{75}
\]
Clearly, $Z\cap(-7Z)=\varnothing.$ 
Thus, the cyclic code $C$ over $\mathbb{F}_{q^2}$ with defining set $Z$ is Hermitian
dual-containing. Using Magma \cite{magma}, we check $d(C)=8$ and $d(C^{\perp_H})\geq22>d(C)$. Therefore, there exists a Hermitian dual-containing optimal cyclic $(21,5)$-LRC
over $\mathbb{F}_{49}$ with parameters
$
[75,60,8]_{49}.
$
Moreover, $d(C^{\perp_H})>d(C)$, the corresponding pure quantum cyclic $(21,5)$-LRC with parameters
$
[[75,45,8]]_{7}
$ is optimal.
\end{example}
This  subsection concludes with Table \ref{tab:hermitian_q2plus1} summarizing the hypotheses and parameters of all the resulting families.
\begin{table*}[ht]
\centering
\caption{Hermitian dual-containing optimal cyclic $(r,\delta)$-LRCs associated
with the $q^2+1$ construction and the corresponding  optimal  pure quantum $(r,\delta)$-LRCs. Throughout the table, $n=u(r+\delta-1)$ and $\delta$ are odd integers and $Z=L\cup D$, where $L=\bigcup_{\substack{1\leq j\leq\delta-2\\ j\ {\rm odd}}}
\{m(r+\delta-1)\pm j:0\leq m\leq u-1\}$}
\label{tab:hermitian_q2plus1}
\renewcommand{\arraystretch}{1.35}
\setlength{\tabcolsep}{4pt}
\scriptsize
\begin{tabular}{|c|p{3.2cm}|p{4cm}|p{3.4cm}|p{4.5cm}|}
\hline
\textbf{Result}
&
\textbf{Conditions}
&
\textbf{Distance Defining set}
&
\textbf{Optimal Classical LRC}
&
\textbf{Optimal Quantum LRC}
\\
\hline

Thm.\ref{thm:q2+1 basic construction}
&

$n\mid(q^2+1)$,

$q(\delta-2)+\ell<r+\delta-1$,

$(q+1)\ell<n$
&

$D=\{\pm1,\pm3,\ldots,\pm\ell\}$ 

where $\ell>\delta-2$ odd,
&
$\bigl[n,\,
ur-\ell+\delta-2,\,
\ell+2\bigr]_{q^2}$
&
$\bigl[\!\bigl[n,\,
u(r-\delta+1)-2\ell+2(\delta-2),\ell+2
\bigr]\!\bigr]_q$
\\
\hline

Thm.\ref{thm:q2+1 two zero extention}
&

$(r+\delta-1)\mid(q^2+1)$,

$\gcd(u,r+\delta-1)\mid\delta$,

$\delta<q$,

$q(\delta-2)+1<r$
&

$D=\{au,-au\}$,

$au\equiv \delta\pmod{r+\delta-1}$
&
$\bigl[n,\,ur-2,\,\delta+2\bigr]_{q^2}$
&
$\bigl[\!\bigl[n,\,
u(r-\delta+1)-4,\delta+2
\bigr]\!\bigr]_q$
\\
\hline

Thm.\ref{thm:q2+1 2delta construction}

&

$(r+\delta-1)\mid(q^2+1)$,

$\gcd(u,r+\delta-1)=1$,

$\gcd(u,q^2-1)\nmid(q+1)$,

$\delta<q$,

$q(\delta-1)<r$
&

$D_1=\{\frac{u}{\zeta}(r+\delta-1)\}$,

$D_2=
\{\pm2auj:
1\leq j\leq(\delta-1)/2\}$,

$\zeta=\gcd(u,q^2-1)$,

$au\equiv 1\pmod{r+\delta-1}$
&
$\bigl[n,\,ur-\delta,\,2\delta\bigr]_{q^2}$
&
$\bigl[\!\bigl[n,\,
u(r-\delta+1)-2\delta,\ 2\delta
\bigr]\!\bigr]_q$
\\
\hline

Cor.\ref{cor:q2+1 2delta-2 construction}
&
$(r+\delta-1)\mid(q^2+1)$,

$\gcd(u,r+\delta-1)=1$,

$\gcd(u,q^2-1)\nmid(q+1)$,

$\delta-2<q$,

$q(\delta-2)-1<r$
&

$D_1=\{\frac{u}{\zeta}(r+\delta-1)\}$,

$D_2=\{\pm2auj:1\leq j\leq(\delta-3)/2\}$

$au\equiv 1\pmod{r+\delta-1}$
&
$\bigl[n,\,
ur-\delta+2,\,
2\delta-2\bigr]_{q^2}$
&
$\bigl[\!\bigl[n,\,
u(r-\delta+1)-2(\delta-2), 2\delta-2
\bigr]\!\bigr]_q$
\\
\hline
\end{tabular}
\end{table*}

\section{Comparison with Closely Related Constructions}\label{sec:comparison}

In this section, we will compare our constructions  with several closely related recent constructions of quantum LRCs. Luo et al. \cite{LuoEtAl2025} constructed
optimal quantum LRCs from CSS codes and, in particular, obtained two cyclic
families for the case $\delta=2$.  Rajpurohit et al. \cite{RajpurohitBhaintwal2026} subsequently constructed quantum
$(r,\delta)$-LRCs from Euclidean dual-containing cyclic codes. Galindo et al. \cite{GalindoHernandoMatsumoto2026BCH} used BCH and
homothetic-BCH codes to obtain optimal pure quantum $(r,\delta)$-LRCs, including
families of unbounded length and families associated with divisors of
$q^2+1$.  Matrix-product constructions
provide further optimal pure families with flexible parameters, but those
codes are generally not cyclic
\cite{CaoZhou2026,GalindoEtAl2026MPC}. The main distinctions are summarized in Table \ref{tab:comparison-previous}.

\begin{table}[htbp]
\centering
\small
\caption{Comparison with closely related constructions of quantum locally recoverable codes.}
\label{tab:comparison-previous}
\begin{tabular}{p{0.19\textwidth}p{0.23\textwidth}p{0.49\textwidth}}
\hline
Reference
& Construction framework
& Comparison with the present work \\
\hline

Luo et al.~\cite{LuoEtAl2025}
& CSS-based cyclic constructions, mainly for $\delta=2$
& For $\delta=2$, our Theorem \ref{thm:basic construction} and Theorem \ref{thm:one zero extension} have the same parameters but removes the restriction $u+2\ell<r+2$ and $u+2<r$, $\gcd(u,q-1)|\frac{2(q-1)}{r+1}$, respectively. Our results also cover arbitrary $\delta\geq2$ and include Hermitian constructions.  \\

\hline

Rajpurohit et al.~\cite{RajpurohitBhaintwal2026}
& CSS constructions with bounded and unbounded length families
& Their unbounded length family has minimum distance $\delta$.
In contrast, our constructions provide optimal pure quantum
$(r,\delta)$-LRCs with minimum distances ranging from $\delta+1$ up to
$2\delta$, depending on the construction. \\

\hline

Galindo et al.~\cite{GalindoHernandoMatsumoto2026BCH}
& BCH and homothetic-BCH constructions using Euclidean and Hermitian duality
& For comparable choices of $n$, $r$, and $\delta$, their explicit optimal
pure families have quantum minimum distance $\delta$.
Our constructions yield additional distance values beyond $\delta$, including
families with $\delta+1$, $\delta+2$, and more generally
$\delta+2\leq d\leq2\delta$.
In the $q^2+1$ case, we also obtain families with
$d=2\delta-2$ and $d=2\delta$. \\

\hline

Matrix-product constructions
\cite{GalindoEtAl2026MPC,CaoZhou2026}
& Noncyclic matrix-product codes based on Euclidean or Hermitian duality
& There is no general parameter domination in either direction. Our contribution is an explicit cyclic defining set realization, whereas the matrix-product papers establish purity for broader noncyclic families.  \\

\hline
\end{tabular}
\end{table}

When
$\delta=2$, our consecutive zero construction Theorem \ref{thm:basic construction} gives
$
\bigl[\!\bigl[u(r+1),u(r-1)-2(\ell-1),\ell+1\bigr]\!\bigr]_q
$
under $u(r+1)\mid(q-1)$ and $2\leq\ell<r$. These are exactly the parameters
of the cyclic family in \cite[Theorem~14]{LuoEtAl2025}, but their construction
requires $u+2\ell<r+2$. Similarly, our one-zero extension construction i.e., Theorem \ref{thm:one zero extension} for $\delta=2$
gives
$
\bigl[\!\bigl[u(r+1),u(r-1)-2,3\bigr]\!\bigr]_q
$
without the conditions $u+2<r$ and
$\gcd(u,q-1)\mid 2(q-1)/(r+1)$ imposed in
\cite[Theorem~15]{LuoEtAl2025}. Thus, even where the parameter formulas
coincide, the present results \textit{enlarge} the admissible parameter region.

The comparison with \cite{RajpurohitBhaintwal2026} is complementary rather than one-sided. Their first family has global distance $r+1$, whereas our
general Euclidean construction Theorem \ref{thm:basic construction} reaches the prescribed distance at most $r$, but optimal with respect to Singleton-type bound \eqref{singleton_qlrc}. On the other hand, their optimal unbounded length
$(r,\delta)$-LRC family has distance $\delta$, whereas our unbounded length constructions Theorem \ref{thm:one zero extension}, \ref{thm:two-zero extension}, \ref{thm:general construction} give distances strictly larger than $\delta$ while
retaining cyclicity and Singleton optimality.

The closest comparison in the $q^2+1$ regime is with
\cite[Theorem 30, 34A]{GalindoHernandoMatsumoto2026BCH}. For the same length
$n$ and locality $(r,\delta)$, the BCH construction gives $
\bigl[\!\bigl[n,k,\delta\bigr]\!\bigr]_q,
$
whereas the present signed-residue constructions in Subsection \ref{subsec:q2-plus-one} give higher-distance points
of the form $\bigl[\!\bigl[n,k',d\bigr]\!\bigr]_q,\ d>\delta.
$
For example, with $q=4$, $u=3$, $r=15$, and $\delta=3$, the BCH
family gives $\bigl[\!\bigl[51,39,3\bigr]\!\bigr]_4$, while our construction by Theorem \ref{thm:q2+1 2delta construction}
gives $\bigl[\!\bigl[51,33,6\bigr]\!\bigr]_4$. These are not strict parameter
improvements because the increase in distance is accompanied by the dimension
loss required by the Singleton-type bound. Rather, they are new optimal points
on the rate--distance trade-off for the same length and local recovery
parameters.

The main novelty of the present work can therefore be stated as follows. First, we develop a cyclic defining set framework for both Euclidean
and Hermitian dual containment and for general $(r,\delta)$, whereas the
cyclic CSS families in \cite{LuoEtAl2025} are restricted to $\delta=2$. Second, it provides multiple
distance extensions and not just the baseline distance $\delta$ for unbounded length code. Third, the signed-residue construction for
$r+\delta-1\mid(q^2+1)$ yields explicit Hermitian dual-containing cyclic families with distances up to $2\delta$. These claims distinguish the present
results from the BCH, homothetic-BCH, and Euclidean dual-containing cyclic families considered in
the related literature.

\section{Conclusion}\label{sec:conclusion}
We developed a defining set framework constructing dual-containing cyclic
$(r,\delta)$-LRCs. A periodic set of zeros guarantees $(r,\delta)$-locality, while carefully chosen additional zeros simultaneously
control the minimum distance and preserve dual containment.

For the Euclidean case, the construction yields dual-containing optimal cyclic
$(r,\delta)$-LRCs over $\mathbb{F}_q$,  including a family with minimum distance $\ell+1$ and further families whose minimum distances range from $\delta+1$ up to $2\delta$. For the Hermitian case with
$r+\delta-1\mid(q^2-1)$, analogous families are obtained over
$\mathbb{F}_{q^2}$. The case $r+\delta-1\mid(q^2+1)$ requires a different
geometry of defining zeros: for odd $n$ and odd $\delta$, symmetric pairs of
odd and even residues are used to exploit the congruence
$q^2\equiv-1\pmod{r+\delta-1}$. This produces further optimal families with
minimum distances $\ell+2,\ \delta+2,\ 2\delta-2,\ 2\delta$. Through the CSS and Hermitian constructions, the classical codes give pure quantum
cyclic $(r,\delta)$-LRCs over $\mathbb{F}_q$ attaining the quantum Singleton-type bound.

Several interesting questions remain open. The divisibility and coprimality hypotheses imposed in our constructions could also be relaxed by allowing
larger cyclotomic cosets or by passing to constacyclic and quasi-cyclic
defining set constructions. In the case $r+\delta-1|(q^2+1)$, the construction problem is not yet fully resolved, particularly when the code length $n$ is even or when $\delta$ is even. Another important direction is to investigate impure quantum $(r,\delta)$-LRCs and determine whether such codes can attain or exceed bounds derived under purity assumptions.

  \section*{Acknowledgments}
A. Agrawal gratefully acknowledges financial support from the
Anusandhan National Research Foundation (ANRF), Government of India, under the National Post Doctoral Fellowship (N-PDF) scheme,
Sanction order No. \textnormal{PDF/2025/000925}. S. S. Garani is supported by an ANRF grant CRG/2022/03469 as part of this work.
  
  \section*{Statements and Declarations}
  \textbf{Author’s Contribution}: All authors contributed equally to this work.\\
  \textbf{Conflict of interest:} The authors have no conflicts of interest in this paper.
 
\footnotesize
\bibliographystyle{abbrv}
\bibliography{references}

\end{document}